\documentclass[11pt]{article}

\usepackage[margin=1in]{geometry}

\usepackage{amssymb,amsmath,amsthm,amsfonts}
\usepackage{bm}
\usepackage{textgreek}
\usepackage{mathtools}
\usepackage{enumitem}
\usepackage[numbers,comma,sort&compress]{natbib}
\usepackage{graphicx}
\usepackage[font=small]{caption}
\usepackage{subcaption} 
\usepackage{float}
\usepackage[ruled,vlined,linesnumbered]{algorithm2e}

\usepackage[pagebackref,colorlinks]{hyperref}

\usepackage{physics}
\usepackage{float}
\usepackage{tablefootnote}
\usepackage{overpic}
\usepackage{multirow}
\usepackage{afterpage}
\usepackage{multirow}     
\usepackage{booktabs}      
\usepackage{tabularx}
\usepackage{makecell}
\usepackage{bbm}

\usepackage{footnotehyper} 
\makesavenoteenv{table}

\newtheorem{theorem}{Theorem}[section]

\newtheorem{lemma}[theorem]{Lemma}
\newtheorem{corollary}[theorem]{Corollary}
\newtheorem{proposition}[theorem]{Proposition}

\newtheorem{definition}[theorem]{Definition}

\newcommand{\eq}[1]{\hyperref[eq:#1]{Equation~(\ref*{eq:#1})}}

\newcommand{\thm}[1]{\hyperref[thm:#1]{Theorem~\ref*{thm:#1}}}
\newcommand{\cor}[1]{\hyperref[cor:#1]{Corollary~\ref*{cor:#1}}}
\newcommand{\defn}[1]{\hyperref[defn:#1]{Definition~\ref*{defn:#1}}}
\newcommand{\lem}[1]{\hyperref[lem:#1]{Lemma~\ref*{lem:#1}}}
\newcommand{\prop}[1]{\hyperref[prop:#1]{Proposition~\ref*{prop:#1}}}
\newcommand{\assum}[1]{\hyperref[assum:#1]{Assumption~\ref*{assum:#1}}}
\newcommand{\fig}[1]{\hyperref[fig:#1]{Figure~\ref*{fig:#1}}}
\newcommand{\tab}[1]{\hyperref[tab:#1]{Table~\ref*{tab:#1}}}
\newcommand{\algo}[1]{\hyperref[algo:#1]{Algorithm~\ref*{algo:#1}}}
\renewcommand{\sec}[1]{\hyperref[sec:#1]{Section~\ref*{sec:#1}}}
\newcommand{\append}[1]{\hyperref[append:#1]{Appendix~\ref*{append:#1}}}
\newcommand{\fac}[1]{\hyperref[fac:#1]{Fact~\ref*{fac:#1}}}
\newcommand{\lin}[1]{\hyperref[lin:#1]{Line~\ref*{lin:#1}}}

\newcommand{\cond}[1]{\hyperref[cond:#1]{Condition~\ref*{cond:#1}}}

\newcommand{\eps}{\varepsilon}

\def\>{\rangle}
\def\<{\langle}

\renewcommand{\i}{\mathrm{i}}

\renewcommand\bra[1]{{\langle{#1}|}}
\renewcommand\ket[1]{{|{#1}\rangle}}
\renewcommand{\ketbra}[2]{\ket{#1}\!\bra{#2}}

\DeclarePairedDelimiter\rbra{\lparen}{\rparen}

\DeclarePairedDelimiter\cbra{\{}{\}}

\DeclarePairedDelimiter\ceil{\lceil}{\rceil}

\newcommand{\TsaF}[2]{\operatorname{F}_{#1}\rbra*{#2}}
\newcommand{\Tsa}[2]{\operatorname{H}_{#1}^{\textup{Tsa}}\rbra*{#2}}
\newcommand{\Ren}[2]{\operatorname{H}_{#1}^{\textup{R\'en}}\rbra*{#2}}

\usepackage{interval}
\intervalconfig{soft open fences}
\allowdisplaybreaks

\SetKwInOut{Input}{Input}
\SetKwInOut{Registers}{Registers}
\SetKwInOut{Output}{Output}
\SetKw{Apply}{Apply}
\SetKw{Construct}{Construct}
\SetKw{Estimate}{Estimate}
\SetKw{Return}{Return}

\title{Quantum Multi-Level Estimation of Functionals of\\Discrete Distributions}
\author{Kean Chen\thanks{\url{keanchen.gan@gmail.com}} \and Minbo Gao\thanks{\url{gmb17@tsinghua.org.cn}} \and Tongyang Li\thanks{\url{tongyangli@pku.edu.cn}} \and Qisheng Wang\thanks{\url{QishengWang1994@gmail.com}} \and Xinzhao Wang\thanks{\url{wangxz@stu.pku.edu.cn}}
}
\date{}

\begin{document}

\maketitle

\begin{abstract}
    We propose a quantum multi-level estimation framework for a functional $\sum_{i=1}^n f(p_i)$ of a discrete distribution $(p_i)_{i=1}^n$. 
    We partition the values $p_i$ into logarithmically many intervals whose length decays exponentially.
    For each interval, we perform non-destructive singular value discrimination to isolate the relevant $p_i$, 
    enabling adaptive estimation of the partial sum over this interval.
    Unlike previous variable-time approaches, our method avoids high control overhead and requires only constant extra ancilla qubits.
    As an application, we present efficient quantum estimators for the $q$-Tsallis entropy of discrete distributions. 
    Specifically, 
    \begin{itemize} 
        \item For $q > 1$, we obtain a \textit{near-optimal} quantum algorithm with query complexity $\widetilde{\Theta}(1/\varepsilon^{\max\{\frac{1}{2(q-1)}, 1\}})$,\footnote{Throughout this paper, $\widetilde{O}\rbra{\cdot}$, $\widetilde{\Omega}\rbra{\cdot}$, and $\widetilde{\Theta}\rbra{\cdot}$ suppress polylogarithmic factors.} improving the prior best $O(1/\varepsilon^{1+\frac{1}{q-1}})$ due to \hyperlink{cite.liu2025on}{Liu and Wang (SODA 2025; \textit{IEEE Trans.\ Inf.\ Theory} 2026)}. 
  
        \item For $0 < q < 1$, we obtain a quantum algorithm with query complexity $\widetilde{O}(n^{\frac{1}{q}-\frac{1}{2}}/\varepsilon^{\frac{1}{q}})$, exhibiting a quantum speedup over the near-optimal classical estimators due to \hyperlink{cite.JVHW17}{Jiao, Venkat, Han, and Weissman (\textit{IEEE Trans.\ Inf.\ Theory} 2017)}. 
    \end{itemize}
    Our results achieve, to our knowledge, the \textit{first} near-optimal quantum estimators for parameterized $q$-entropy for non-integer $q$.
    
\end{abstract}


\newpage
\tableofcontents
\newpage

\section{Introduction}

Property testing on quantum computers is a growing field in quantum computing \cite{MdW16}. 
An important question is whether quantum computing can bring speedups to classical tasks, namely the property testing of discrete distributions. 
Early works \cite{bravyi2011quantum,chakraborty2010new} investigated the quantum query complexity of testing a series of properties of discrete probability distributions, including identity, closeness, uniformity, and orthogonality.
A quantum algorithm for the hypothesis testing of discrete distributions was presented in \cite{belovs2019quantum}. 

Quantum approaches for parameterized tasks regarding discrete distributions were first considered in \cite{li2018quantum}. 
Specifically, they presented quantum estimators for the $\alpha$-R\'enyi entropy $\Ren{\alpha}{p} = \frac{1}{1-\alpha} \log ( \sum_{i=1}^n p_i^\alpha )$ \cite{Ren61} for an unknown discrete distribution $p = (p_i)_{i=1}^n$, where $\alpha \in \interval{0}{+\infty}$ is the parameter. 
In particular, their quantum estimators for the $1$- and $2$-R\'enyi entropies achieve near-optimal dependence on $n$ (cf.\ \cite{BKT20}). 
Here, the $1$-R\'enyi entropy is the Shannon entropy \cite{Sha48a,Sha48b}. 
Recently, an improved quantum estimator for Shannon entropy was proposed in \cite{shin2025near}. 
Another parameterized entropy of great interest is the $q$-Tsallis entropy $\Tsa{q}{p} = \frac{1}{1-q} ( \sum_{i=1}^n p_i^q - 1 )$ \cite{Tsa88}, for which near-optimal quantum estimators were provided in \cite{Wan25} for any integer $q \geq 2$.
For non-integer $q > 2$, quantum algorithms with query complexity $O(1/\varepsilon^{1+\frac{1}{q-1}})$ were proposed in \cite{liu2025on}, improving the classical sample complexity $O(1/\varepsilon^2)$ given in \cite{JVHW17}, where $\varepsilon$ is the desired additive error. 

Another parameterized task is the $\ell_\alpha$ closeness testing, where $\alpha \in \interval{1}{+\infty}$ is the parameter. 
The quantum query complexity of the $\ell_1$ closeness testing was first studied in \cite{bravyi2011quantum}, where they further considered the estimation of the $\ell_1$ distance. 
Improved quantum algorithms for the $\ell_1$ closeness testing were subsequently proposed in \cite{Mon15,gilyen2020distributional,LWL24,CKO25}.
The $\ell_2$ closeness testing was considered in \cite{gilyen2019quantum} and later improved to optimal up to constant factors in \cite{LWL24}. 
For non-integer $\alpha > 1$, there is only an attempt in \cite{LW25Lalpha} for the $\ell_\alpha$ closeness testing of quantum states. 

As we can see, (near-)optimal quantum estimators/testers for discrete distributions are only known for integer parameters. 
We therefore ask the following question:
\[
\textit{Can we find optimal quantum estimators for properties with non-integer parameters?}
\]
In this paper, we answer this question in the affirmative.
We propose a multi-level estimation framework that provides a comprehensive resolution to the problem by establishing near-optimal quantum estimators for the $q$-Tsallis entropy for all non-integers $q > 1$.
A key feature of our framework is that it addresses the implementation overheads of prior art while achieving a fine-grained query complexity.
Specifically, we implement this multi-level structure using only a constant number of additional ancilla qubits, effectively eliminating the complex control logic required by previous methods.

The necessity of such a multi-level strategy arises from a primary technical challenge: the component function $f(x)$ in the target functional $\sum_{i=1}^n f(p_i)$ often exhibits a singular behavior at zero. 
Important examples include the Shannon entropy and the $q$-Tsallis entropy for non-integer $q$.
The standard QSVT-based framework~\cite{gilyen2020distributional} typically uses a \emph{single} polynomial to approximate the target function $f(x)$ globally.
However, in the presence of singularities, efficient polynomial approximation is typically feasible only on an interval $[\delta, 1]$ bounded away from zero, where the required degree scales inversely with the cutoff $\delta$.

A similar difficulty appears in the quantum linear system problem where the inverse function is singular at the origin.
A standard solution in that context is Variable Time Amplitude Amplification (VTAA)~\cite{ambainis2012variable,childs2017quantum}, which partitions the state space based on the distance to the singularity.
While improving query complexity, this strategy inherently incurs significant implementation overheads: partitioning the state space often requires complex multi-qubit controlled operations and a large number of ancilla qubits.
Prior approaches adapting this partitioning idea to distribution testing~\cite{wang2024quantum,shin2025near} inherit these heavy circuit complications, which our framework successfully resolves.

\subsection{Main results}

In this paper, to estimate functionals of an unknown discrete distribution $(p_i)_{i=1}^n$, we employ the \textit{purified quantum query access} model, where one is given a unitary oracle, $\mathcal{O}_p$, for the distribution $p$ such that
\begin{equation*}
    \mathcal{O}_p\ket{\mathbf{0}}\ket{\mathbf{0}} = \sum_{i=1}^n \sqrt{p_i}\ket{i}\ket{\phi_i},
\end{equation*}
with each $\ket{\phi_i}$ an arbitrary normalized pure state. 
This quantum query model is the now standard input model for quantum property testing, commonly used in quantum computational complexity \cite{Wat02} and quantum algorithms \cite{gilyen2020distributional}. 

To establish our quantum estimators, we provide a quantum multi-level estimation framework for general functionals of the form $\sum_{i=1}^n f(p_i)$. While building on the partitioning strategy of \cite{shin2025near}, our framework significantly reduces the circuit implementation cost. We introduce the function $g(x) \coloneqq f(x)/x$ to rewrite the target functional as $\sum_{i=1}^n p_i g(p_i)$.

\begin{theorem}[Informal version of \thm{multi-level-ae-1}]\label{thm:framework-intro}
    Suppose $(p_i)_{i=1}^n$ is an unknown discrete distribution. 
    We partition $\interval{0}{1}$ 
    into intervals $\Delta_{m+1}=[0, \varphi_m], \Delta_{m}=(\varphi_{m}, \varphi_{m-1}],\ldots, \Delta_1=(\varphi_1, 1]$, with geometrically decreasing $\varphi_j$.
    For $1 \leq j \leq m$, 
    let $n_j$ be the number of $\sqrt{p_i}$ falling into $\Delta_j$, and let $\mathcal{B}_j$ be an upper bound on $2|g(x^2)|$ over all $x$ in $\Delta_j$ and its adjacent intervals.
    
    If $g(x^2)$ can be approximated by the square of a polynomial of degree $\widetilde{O}(1/\varphi_j)$ on $\Delta_j$ for each $1\le j\le m$ (see \cond{approx-P-j}) and $\sum_{\sqrt{p_i} \in \Delta_{m+1}} p_i \abs{g(p_i)}$ is bounded by $O\rbra{\varepsilon}$ (see \cond{tail-bound}), then
    there is a quantum algorithm that 
    estimates the functional $\sum_{i=1}^n p_i g(p_i)$ to within additive error $\eps$ with high probability, using
    \begin{align*}
        \widetilde{O}\left( \sum_{j=1}^m \left( \frac{\mathcal{B}_j \sqrt{n_j}}{\eps} + \frac{\sqrt{\mathcal{B}_j}}{\varphi_{j}\sqrt{\eps}} \right) \right)
    \end{align*}
    queries to the purified quantum query access oracle $\mathcal{O}_p$ and only $4$ additional qubits beyond the register required for the projected unitary encoding of the distribution $p$. 
\end{theorem}

As a direct application of \thm{framework-intro}, we obtain a quantum estimator for Shannon entropy with query complexity $O(\sqrt{n}\log^{4.5}(n/\eps)/\varepsilon)$, reproducing the result of \cite{shin2025near} with explicit polylogarithmic factors (see \cor{Shannon-entropy}). 

Our specific main applications of \thm{framework-intro} are a series of quantum estimators for the $q$-Tsallis entropy~\cite{Tsa88}
\[
\Tsa{q}{p} = \frac{1}{1-q} \left( \sum_{i=1}^n p_i^q - 1 \right)
\]
for all $q \in \interval[open]{0}{1} \cup \interval[open]{1}{+\infty}$, where $\TsaF{q}{p} = \sum_{i=1}^n p_i^q$ is the key component to estimate. 
Classical estimators for the $q$-Tsallis entropy have been investigated in \cite{AK01,JVHW15,JVHW17}.
We introduce our quantum estimators below and compare them with prior results in \tab{quantum-query-complexity-tsallis}.

\begin{table}[tb]
\centering
\begin{tabular}{|c|cc|c|cc|}
\hline
 &
  \multicolumn{1}{c|}{$0<q<0.5$} &
  $0.5\le q < 1$ &
  $1<q<1.5$ &
  \multicolumn{1}{c|}{$1.5\le q< 2$} &
  $q > 2$ \\ \hline
\begin{tabular}[c]{@{}c@{}}Previous\\ Upper Bounds\end{tabular}&
  \multicolumn{1}{c|}{\begin{tabular}[c]{@{}c@{}}$\rule{0pt}{20pt}\widetilde{O}\rbra*{\frac{n^{\frac{1}{q}}}{\varepsilon^{\frac{1}{q}}}}$\\ \cite{JVHW15}\end{tabular}} &
  \begin{tabular}[c]{@{}c@{}}\rule{0pt}{20pt}$\widetilde{O}\rbra*{\frac{n^{\frac{1}{q}}}{\varepsilon^{\frac{1}{q}}}+\frac{n^{2-2q}}{\varepsilon^2}}$\\ \cite{JVHW15}\end{tabular} &
  \begin{tabular}[c]{@{}c@{}}\rule{0pt}{20pt}$\widetilde{O}\rbra*{\frac{1}{\varepsilon^{\frac{1}{q-1}}}}$\\ \cite{JVHW15}\end{tabular} &
  \multicolumn{1}{c|}{\begin{tabular}[c]{@{}c@{}}$O\rbra*{\frac{1}{\varepsilon^2}}$\\ \cite{JVHW17}\end{tabular}} &
  \begin{tabular}[c]{@{}c@{}}\rule{0pt}{20pt}$O\rbra*{\frac{1}{\varepsilon^{1+\frac{1}{q-1}}}}$\\ \cite{liu2025on}\end{tabular} \\ \hline
\begin{tabular}[c]{@{}c@{}}Our \\ Upper Bounds\end{tabular} &
  \multicolumn{2}{c|}{\begin{tabular}[c]{@{}c@{}}\rule{0pt}{20pt}$\widetilde{O}\left(\frac{n^{\frac{1}{q}-\frac{1}{2}}}{\varepsilon^{\frac{1}{q}}}\right)$\\
  \thm{1312127}
  \end{tabular}} &
  \multirow{2}{*}{\begin{tabular}[c]{@{}c@{}}\rule{0pt}{10pt}$\widetilde{\Theta}\left(  \frac{1}{\varepsilon^{\frac{1}{2(q-1)}}}\right)$\\
    \thm{Tsallis-entropy-q-greater-than-one}\\
  \thm{lower-bound-q-tsallis}
  \end{tabular}} &
  \multicolumn{2}{c|}{\multirow{2}{*}{\begin{tabular}[c]{@{}c@{}}\rule{0pt}{20pt}$\widetilde{\Theta}\rbra*{\frac{1}{\varepsilon}}$\\
  \thm{Tsallis-entropy-q-greater-than-one}\\
  \thm{lower-bound-q-tsallis}
  \end{tabular}}} \\ \cline{1-3}
\begin{tabular}[c]{@{}c@{}}Our\\ Lower Bounds\end{tabular} &
  \multicolumn{1}{c|}{\begin{tabular}[c]{@{}c@{}}\rule{0pt}{20pt}$\widetilde{\Omega}\left(\frac{n^{\frac{1}{2q}}}{\varepsilon^{\frac{1}{2q}}}\right)$\footnote{This lower bound holds for $\varepsilon \geq \Omega\!\left(n^{\frac{2q^2-2q+1}{1-2q}}\log^{\frac{q}{1-2q}}(n)\right)$.} \\
  \thm{lower-bound-q-tsallis}
  \end{tabular}} &
  \begin{tabular}[c]{@{}c@{}}\rule{0pt}{20pt}$\widetilde{\Omega}\left(\frac{n^{\frac{1}{2q}}}{\varepsilon^{\frac{1}{2q}}} +\frac{n^{1-q}}{\varepsilon}\right)$ \\
  \thm{lower-bound-q-tsallis}
  \end{tabular} &
   &
  \multicolumn{2}{c|}{} \\ \hline
\end{tabular}
\caption{Quantum query complexity of estimating $q$-Tsallis entropy for non-integer $q$.}
\label{tab:quantum-query-complexity-tsallis}
\end{table}

\begin{theorem}[Tsallis entropy estimation for $q\ge 1.5$, \thm{Tsallis-entropy-q-greater-than-one} and \thm{lower-bound-q-tsallis} combined] \label{thm:Tsallis-geq-1.5-intro}
    For any constant $q \ge 1.5$, the $q$-Tsallis entropy $\Tsa{q}{p}$ can be estimated to within additive error $\varepsilon$, with near-optimal quantum query complexity \[ \widetilde{\Theta}\left( \frac{1}{\varepsilon}\right).\]
\end{theorem}

\begin{theorem}[Tsallis entropy estimation for $1<q<1.5$, \thm{Tsallis-entropy-q-greater-than-one} and \thm{lower-bound-q-tsallis} combined] \label{thm:Tsallis-1-to-1.5-intro}
    For any constant $1<q<1.5$, 
    the $q$-Tsallis entropy $\Tsa{q}{p}$ can be estimated to within additive error $\varepsilon$, with near-optimal quantum query complexity \[\widetilde{\Theta}\left(  \frac{1}{\varepsilon^{\frac{1}{2(q-1)}}}\right).\]
\end{theorem}

To the best of our knowledge, \thm{Tsallis-geq-1.5-intro} and \thm{Tsallis-1-to-1.5-intro} present the first near-optimal quantum estimators for the (parameterized) $q$-entropy of probability distributions for \textit{non-integer} $q$.\footnote{Near-optimal quantum estimators for the $q$-entropy of probability distributions are known for integer $q$ in certain cases, e.g., the $2$-R\'enyi entropy ($q = 2$) \cite{li2018quantum}, the $q$-Tsallis entropy for integer $q\ge 2$ \cite{Wan25}, and the Shannon entropy ($q = 1$) \cite{shin2025near}.} 
It is known that the sample complexity of estimating the $q$-Tsallis entropy is $\Theta(1/\varepsilon^2)$ for $q \geq 1.5$ \cite{JVHW17} and $\widetilde{\Theta}(1/\varepsilon^{\frac{1}{q-1}})$ for $1 < q < 1.5$ \cite{JVHW15}. 
In comparison, \thm{Tsallis-geq-1.5-intro} and \thm{Tsallis-1-to-1.5-intro} achieve a strict quantum advantage over the corresponding classical estimators. 
Moreover, our results also improve the previous best query complexity $O(1/\varepsilon^{1+\frac{1}{q-1}})$ for $q > 1.5$ due to \cite{liu2025on}.\footnote{In \cite{liu2025on}, they considered the estimation of the $q$-Tsallis entropy of quantum states, which is a more general task. Even though, their query complexity remains unchanged when the quantum states degenerate to probability distributions.} 

\begin{theorem}[Tsallis entropy estimation for $0<q<1$, \thm{1312127} restated]\label{thm:232230}
For a constant $0<q<1$, and a probability distribution $p$ of size $n$, the $q$-Tsallis entropy $\Tsa{q}{p}$ can be estimated to within additive error $\varepsilon$, with quantum query complexity 
\[\widetilde{O}\left(\frac{n^{\frac{1}{q}-\frac{1}{2}}}{\varepsilon^{\frac{1}{q}}}\right).\]
\end{theorem}

\thm{232230} demonstrates a quantum speedup over the near-optimal classical estimator~\cite{JVHW15} for $q$-Tsallis entropy with sample complexity $\widetilde{\Theta}\!\left(\frac{n^{1/q}}{\varepsilon^{1/q}}+\frac{n^{2-2q}}{\varepsilon^2}\right)$ for $0<q<1$.
We also present a quantum query lower bound $\widetilde{\Omega}(n^{1/2q})$ (see \thm{lower-bound-q-tsallis} for details), suggesting that there is only room for sub-quadratic speedup on the distribution size $n$.
We note that when $q=1-o(1)$, \thm{232230} achieves optimal dependence on $n$ up to a quasi-polynomial factor.

Moreover, given that an estimate of $q$-Tsallis entropy is always an estimate of $q$-R\'enyi entropy to the same precision when $0 < q < 1$, \thm{232230} also implies an improved quantum query complexity of estimating the $\alpha$-R\'enyi entropy in certain regimes, combined with the results in {\cite[Theorem 1]{wang2024quantum}}. 

\begin{corollary} \label{cor:renyi-intro}
For a constant $0<\alpha<1$, and a probability distribution $p$ of size $n$, the $\alpha$-R\'enyi entropy $\Ren{\alpha}{p} = \frac{1}{1-\alpha} \log \left( \sum_{i=1}^n p_i^\alpha  \right)$ can be estimated to within additive error $\varepsilon$, with quantum query complexity 
\[\widetilde{O}\left(\min\left\{\frac{n^{\frac{1}{2\alpha}}}{\varepsilon^{\frac{1}{2\alpha}+1}}, \frac{n^{\frac{1}{\alpha}-\frac{1}{2}}}{\varepsilon^{\frac{1}{\alpha}}}\right\}\right).\]
\end{corollary}

\cor{renyi-intro} improves the quantum query complexity $\widetilde{O}\!\left(\frac{n^{\frac{1}{2\alpha}}}{\varepsilon^{\frac{1}{2\alpha}+1}}\right)$ in \cite{wang2024quantum} when $\varepsilon\leq O(n^{\frac{1-\alpha}{1-2\alpha}})$ and $0.5< \alpha < 1$.

\subsection{Techniques}

\paragraph{Limitations of prior approaches.} Given access to the distribution $p$ via the oracle $\mathcal{O}_p$, \cite{gilyen2020distributional} constructed a projected unitary encoding of a matrix $A$ with singular value decomposition $A = \sum_{i=1}^n \sqrt{p_i}\ketbra{\tilde{\psi}_i}{\psi_i}$. Conventional algorithms for estimating the functional $\sum_{i=1}^n p_i g(p_i)$ rely on a single polynomial $P$ such that $P^2(\sqrt{p_i})\approx g(p_i)$ over the entire domain. Applying the quantum singular value transformation (QSVT) of $A$ with the polynomial $P$ to the state $\sum_{i=1}^n \sqrt{p_i}\ket{\psi_i}$ yields a state with squared norm $\sum_{i=1}^n p_i P^2(\sqrt{p_i})\approx \sum_{i=1}^n p_i g(p_i)$, which can be estimated via amplitude estimation. This approach often yields suboptimal query complexity for functions with singularities, such as various entropy measures, where the local behavior of $g(x)$ varies significantly.
For example, when estimating the $q$-Tsallis entropy $\Tsa{q}{p}$, the function $g\rbra{x} = x^{q-1}$ is generally not smooth for any non-integer $q$.

To address this limitation, \cite{wang2024quantum} proposed a framework based on variable-time amplitude estimation (VTAE)~\cite{chakraborty2019power}. By approximating $g(p_i)$ with distinct polynomials across different intervals, the complexity depends on the average polynomial degree rather than purely on the maximum one. Recently, \cite{shin2025near} achieved near-optimal complexity for Shannon entropy estimation ($g(x) = -\log(x)$) using a different approach. This algorithm shares the same spirit as VTAE by approximating $-\log(x)$ with different polynomials on different intervals, but avoids the nested amplitude amplifications used in VTAE. Instead, it uses QSVT with threshold polynomials to partition the values $\sqrt{p_i}$ into $m$ intervals, and then applies amplitude estimation to the contribution of each interval separately.

While enabling fine-grained polynomial approximation, both algorithms in \cite{wang2024quantum,shin2025near} suffer from the high control and ancilla overhead inherent to variable-time quantum algorithms~\cite{ambainis2012variable}. Partitioning the values $\sqrt{p_i}$ into $m$ intervals requires applying a sequence of $m$ QSVT-based thresholding steps. Specifically, the $j$-th thresholding step is designed to select those $\sqrt{p_i}$ lying above a threshold such as $2^{-j}$, and the outcome is stored in an ancilla qubit. To ensure that each $\sqrt{p_i}$ is counted exactly once, the $j$-th unitary is applied conditioned on the ancilla qubits from all preceding steps, so that it acts only on the subspace rejected by the previous $j-1$ steps. This creates a deep sequence of multi-qubit controlled operations, imposing a heavy control overhead. Moreover, a state rejected by the thresholding is mapped to a garbage state, which cannot be reused by the subsequent thresholding step. Consequently, a fresh copy of the input state is required for every thresholding step, leading to a large ancilla overhead. 

\paragraph{Our approach: a resource-efficient multi-level framework.} In \algo{multi-level-ae}, we combine two ideas to obtain query-efficient algorithms while reducing control and ancilla overhead. First, we partition the values $\sqrt{p_i}$ into $m$ groups, such that applying only \emph{two} QSVT-based thresholding steps, corresponding to the lower and upper boundaries, suffices to isolate the values in any single group. A key challenge is that due to the discontinuity of the threshold function, only a ``soft'' thresholding can be implemented, leading to an uncertain transition region. This stands in contrast to arithmetic-based methods, such as those used in quantum mean estimation~\cite{Mon15, hamoudi2019quantum, hamoudi2021quantum,cornelissen2022near}, where values stored in quantum registers can be partitioned exactly using coherent comparators. To address this, we interpret this soft thresholding as assigning a probability weight to each value for every group. In \thm{multi-level-ae-1}, we show that, up to a tolerable error, \begin{enumerate}
    \item (\emph{Localization}) the values in each group lie within an interval where the ratio of the maximum to the minimum value is bounded by a constant, and
    \item (\emph{Completeness}) the probability weights of $\sqrt{p_i}$ sum to one over all groups.
\end{enumerate} The first condition enables us to use a polynomial that approximates the target function for $\sqrt{p_i}$ only within a local region, which allows for a significantly lower polynomial degree. The second condition ensures that the sum of the estimates over all groups correctly approximates the functional. Second, we apply the gapped phase estimation with branch marking technique in \cite{low2024quantum} to implement the thresholding non-destructively. This allows the initial state to be preserved rather than collapsing into a garbage state as in prior thresholding implementations. Consequently, the state can be reused for the second thresholding step, reducing the additional ancilla requirement to a constant. 

To apply branch marking, which requires a Hermitian input operator, we construct a projected unitary encoding of a Hermitian matrix $H$ with spectral decomposition 
\[
H = \sum_{i=1}^n \frac{\sqrt{p_i}}{2}\big(\ketbra{\phi_{i,+1}}{\phi_{i,+1}}- \ketbra{\phi_{i,-1}}{\phi_{i,-1}}\big).
\]
We can prepare the initial state \[
\sum_{i=1}^n \sqrt{\frac{p_i}{2}}\big(\ket{\phi_{i,+1}}+\ket{\phi_{i,-1}}\big)
\]
using one query to $\mathcal{O}_p$. For the $j$-th group, we implement two thresholding steps to isolate the target values using constant ancilla qubits. Given the local property from the first condition, we apply the QSVT of $H$ with a polynomial $P_j$ such that $P_j^2(\sqrt{p_i}/2) \approx g(p_i)$ for $\sqrt{p_i}$ in the corresponding region. As QSVT polynomials have definite parity, we have $P_j^2(-\sqrt{p_i}/2) = P_j^2(\sqrt{p_i}/2) \approx g(p_i)$, which ensures that both positive and negative eigenvalues contribute identically to the final estimate.

For clarity, we provide a high-level overview of our approach in \algo{informal-multi-level}.  A complete description of the procedure is deferred to \algo{multi-level-ae} in \sec{main-algo}.
\begin{algorithm}[htb]
\caption{Informal Overview of Multi-Level Estimation (see \algo{multi-level-ae} for the formal procedure)}
\label{algo:informal-multi-level}
\Input{
    Unitary oracle $\mathcal{O}_p$ for distribution $p$; target function $g$; precision $\eps$; 
    decreasing ratio $\rho$, number of levels $m$, and local bounds $\{B_j\}$.
}
\Output{
    Estimate of $\sum_{i=1}^n p_i g(p_i)$.
}
\BlankLine
$S \leftarrow 0$\;
Define thresholds $\varphi_j = \rho^{-j}$ for $j = 0, \dots, m+1$, which partition $[0, 1]$ into intervals $(\varphi_1, \varphi_0], \dots, (\varphi_m, \varphi_{m-1}]$ and a residual interval $[0, \varphi_m]$\;

\For{$j \leftarrow 1$ \KwTo $m$}{
    Apply non-destructive singular value discriminators to coherently isolate the amplitudes $\sqrt{p_i} \in (\varphi_{j+1}, \varphi_{j-1}]$ into a target subspace\;
    
    Apply QSVT with a local polynomial $P_j$ satisfying $P_j^2(\sqrt{p_i}/2) \approx g(p_i)$ for amplitudes within the interval $(\varphi_{j+1}, \varphi_{j-1}]$\;
    
    Perform quantum amplitude estimation to measure the squared norm of the state in the target subspace, obtaining the partial estimate $\tilde{v}_j$\;
    
    $S \leftarrow S + \max\{B_j, B_{j+1}\} \tilde{v}_j$, where $B_j$ is the upper bound of $|g(x^2)|$ for $x\in(\varphi_j, \varphi_{j-1}]$\;
}
\Return{$S$}\;
\end{algorithm}

\subsection{Open questions}
Our paper leaves several natural open questions for future investigation:
\begin{itemize}
\item The bounds we obtain for Tsallis entropy estimation in the regime 
$0<q<1$ are not optimal. Is it possible to derive improved upper or lower bounds on the corresponding quantum query complexity?
For example, can we show better lower bounds using other lower bound techniques such as the polynomial method~\cite{BBC+01}, adversary method~\cite{Amb02}, and compressed oracle method~\cite{Zha19}?
More broadly, can the idea of multi-level estimation be applied to design optimal quantum algorithms for $\ell_{\alpha}$ closeness testing or for estimating other notions of relative entropy between discrete distributions?

\item There are results on the complexity of entropy estimation of quantum states~\cite{AISW20,gilyen2020distributional,GHS21,SH21,wang2024quantum,wang2024new,WZ25b,liu2025on,CW25}.
The technique of~\cite{gilyen2020distributional} can be used to reduce the estimation of quantum entropies to the estimation of entropies of discrete probability distributions~\cite{wang2024quantum}. 
It remains an open question whether our quantum multi-level estimation technique can be applied in this setting, or adapted appropriately, to obtain optimal algorithms for estimating quantum entropies.
A potential difficulty is that the eigenbasis of the quantum state is unknown.

\item Our multi-level amplitude estimation shares several conceptual similarities with variable-time amplitude amplification and estimation algorithms~\cite{ambainis2012variable, chakraborty2019power,AKV23,low2024quantum}. Are there deeper or more formal connections between these approaches? Moreover, given the broad range of applications of variable-time amplitude amplification, it is natural to ask whether our multi-level amplitude estimation framework can be extended to, or applied in, more general settings.

\end{itemize}

\section{Preliminaries}
\subsection{Notations}
Unless otherwise stated, all operator approximations are measured in the spectral norm $\|\cdot\|$, and state approximations are measured in the Euclidean norm $\|\cdot\|_2$. We use $\ket{\mathbf{0}}$ to denote the all-zero state $\ket{0}^{\otimes k}$ on a quantum register, where the number of qubits $k$ is clear from the context.

\subsection{Quantum subroutines}

In this subsection, we review the basic quantum subroutines that serve as building blocks for our algorithms.

\begin{theorem}[Amplitude estimation \cite{BHMT02}]
\label{thm:ae}
    There is a quantum algorithm that takes as input a unitary $U$, an orthogonal projector $\Pi$, and a positive integer $t$. The algorithm outputs an estimate $\tilde{a}$ of $a = \|\Pi U\ket{0}\|^2$, such that \begin{align*}
        |\tilde{a}-a|\le \frac{2\pi\sqrt{a(1-a)}}{t} + \frac{\pi^2}{t^2}
    \end{align*}
    with probability at least $8/\pi^2$, using $2t+1$ calls to (controlled) $U$ and $U^{\dagger}$, and $t$ calls to the reflection $I-2\Pi$.
\end{theorem}
To estimate $a$ to within a given additive error $\eps$, we propose the following two-stage algorithm that first decides whether $a \le \eps$ and then determines $t$ accordingly. 
\begin{corollary}
    \label{cor:two-stage-ae}
    Let $U$ be a unitary and $\Pi$ be an orthogonal projector. Let $a = \|\Pi U\ket{0}\|^2$. For any precision $\eps \in(0,1]$ and failure probability $\eta \in (0, 1)$, there exists a quantum algorithm that outputs an estimate $\tilde{a}$ satisfying $|\tilde{a}-a| \le \eps$ with probability at least $1-\eta$, using 
    \begin{align*}
        O\left( \left(\frac{\sqrt{a}}{\eps} + \frac{1}{\sqrt{\eps}}\right) \log(\frac{1}{\eta}) \right)
    \end{align*}
    queries to (controlled) $U$, $U^\dagger$, and $I-2\Pi$.
\end{corollary}
\begin{proof}
    Applying the amplitude estimation algorithm in \thm{ae} with $t_1 =\pi \sqrt{80/\eps}$, the output $\tilde{a}_1$ satisfies \begin{align}
    \label{eq:error-ae-1}
        |\tilde{a}_1-a|\le 2\sqrt{\frac{\eps}{80}}\sqrt{a}+ \frac{\eps}{80} = \frac{\sqrt{\eps}}{2\sqrt{5}}\sqrt{a}+\frac{\eps}{80},
    \end{align}
    with probability at least $8/\pi^2$. By repeatedly running $O(\log(1/\eta))$ times and taking the median of the outputs, the success probability can be boosted to $1-\eta/2$.
    
    If $\tilde{a}_1 \le 3\eps/4$, \eq{error-ae-1} implies \begin{align*}
        g(a):= a-\frac{3\eps}{4}-\frac{\sqrt{\eps}}{2\sqrt{5}}\sqrt{a}-\frac{\eps}{80} \le 0.
    \end{align*}
    Observe that $g(a)$ is monotonically decreasing for $a\in [0, \eps/80]$ and increasing for $a\in [\eps/80, \infty)$, and \begin{align*}
        g(\eps) = \frac{19\eps}{80}-\frac{\eps}{2\sqrt{5}} = \frac{(19-8\sqrt{5})}{80}\eps > 0.
    \end{align*}
    Consequently, we have $g(a)\ge 0$ for all $a \ge \eps$. Combining this with the condition $g(a) \le 0$, we conclude that $a < \eps$. Then, we output $0$, which is an $\eps$-approximation of $a$. 
    
    If $\tilde{a}_1 > 3\eps/4$, \eq{error-ae-1} implies \begin{align*}
        h(a):=a - \frac{3\eps}{4} + \frac{\sqrt{\eps}}{2\sqrt{5}}\sqrt{a}+\frac{\eps}{80} >0.
    \end{align*} Since $h(a)$ is monotonically increasing for $a \ge 0$ and \begin{align*}
        h\left(\frac{\eps}{2}\right) = -\frac{19\eps}{80}+\frac{\eps}{2\sqrt{10}} = \frac{(4\sqrt{10}-19)}{80}\eps <0,
    \end{align*}
    we have $h(a) <0$ for all $a \le \eps/2$. Combining this with the condition $h(a) > 0$, we conclude that $a > \eps/2$. Then, \eq{error-ae-1} implies \begin{align*}
        |\tilde{a}_1-a|\le \frac{\sqrt{\eps}}{2\sqrt{5}}\sqrt{a}+\frac{\eps}{80} < \frac{1}{\sqrt{10}} a+ \frac{1}{40}a < \frac{1}{2}a.
    \end{align*}
    This implies $0.5a < \tilde{a}_1 < 1.5a$.
    Finally, we run the amplitude estimation algorithm again with success probability boosted to $1-\eta/2$ using \begin{align*}
        t_2 = \max\left\{\frac{4\pi\sqrt{2\tilde{a}_1}}{\eps} , \sqrt{\frac{2}{\eps}}\pi\right\} = O\left(\frac{\sqrt{a}}{\eps} +\frac{1}{\sqrt{\eps}}\right)
    \end{align*} steps. The output $\tilde{a}_2$ satisfies \begin{align*}
        |\tilde{a}_2-a|\le 2\pi \frac{\sqrt{a}}{t_2} + \frac{\pi^2}{t_2^2}\le  \sqrt{\frac{a}{8\tilde{a}_1}} \eps + \frac{\eps}{2} <  \sqrt{\frac{a}{4a}} \eps + \frac{\eps}{2} =\eps. 
    \end{align*}
    By the union bound, the overall success probability is at least $1-\eta$.
\end{proof}

\subsection{Quantum singular value transformation}

A \emph{projected unitary encoding} $(\widetilde{\Pi}, \Pi, U)$~\cite{gilyen2019quantum} of a matrix $A\in \mathbb{C}^{d\times d}$ is a unitary $U$ and two orthogonal projections $\Pi, \widetilde{\Pi}$ such that \begin{align*}
    A = \widetilde{\Pi} U \Pi.
\end{align*} Let \begin{align*}
    A = \sum_{i = 1}^{d} \sigma_i \ketbra{\tilde{\psi_i}}{\psi_i}
\end{align*} be the \emph{singular value decomposition} (SVD) of $A$. Given a scalar function $f$, the singular value transformation of $A$ is defined by applying $f$ to its singular values.  \begin{definition}[Singular value transformation]
    Let $f:\mathbb{R}\to \mathbb{C}$ be an even or odd function and $A\in \mathbb{C}^{d\times d}$ be a matrix with SVD $ A = \sum_{i = 1}^{d} \sigma_i \ketbra{\tilde{\psi_i}}{\psi_i}$. Define the singular value transformation of $A$ corresponding to $f$ as \begin{align*}
        f^{(\mathrm{SV})}(A) := \begin{cases}
            \sum_{i=1}^d f(\sigma_i) \ketbra{\tilde{\psi_i}}{\psi_i} & \text{ for odd } f, \text{ and}\\
            \sum_{i=1}^d f(\sigma_i) \ketbra{\psi_i}{\psi_i}& \text{ for even }f.
        \end{cases}
    \end{align*}
\end{definition}
Given a projected unitary encoding of $A$, we can implement the singular value transformation of $A$ corresponding to a bounded polynomial with definite parity. 
\begin{theorem}[{\cite[Corollary 11]{gilyen2019quantum}}]
\label{thm:qsvt}
    Suppose that $P\in \mathbb{R}[x]$ is a degree-$m$ polynomial satisfying \begin{itemize}
        \item $P$ has parity-$(m \mod 2)$ and 
        \item for all $x\in[-1,1]$: $|P(x)|\le 1$.
    \end{itemize}
    Then there exists a unitary $U^{(\mathrm{SV})}_P$ such that \begin{align}
    P^{(\mathrm{SV})}\big(\widetilde{\Pi} U \Pi\big) \nonumber 
    =  \begin{cases}
    \big(\langle+| \otimes \widetilde{\Pi}\big) U^{(\mathrm{SV})}_P \big(|+\rangle \otimes \Pi\big) & \text { if } m \text { is odd,} \\
    \big(\langle+| \otimes \Pi\big) U^{(\mathrm{SV})}_P \big(|+\rangle \otimes \Pi\big) & \text { if } m \text { is even. }\end{cases}
  \end{align}
    The unitary $U^{(\mathrm{SV})}_P$ can be implemented using $O(m)$ elementary gates, controlled reflections $I-2\Pi, I-2\widetilde{\Pi}$, and $O(m)$ queries to controlled-$U$ and $U^{\dagger}$.
\end{theorem}
\subsection{Polynomial approximations}

We recall the following lemmas about polynomial approximations of power functions.

\begin{lemma}[Polynomial Approximations of Negative Power Functions,~\cite{gilyen2019quantum}]%
\label{lem:poly-approx-negative-powers}
    For $\delta, \varepsilon\in \interval[open left]{0}{\frac{1}{2}}$ and $c >0$, let $g(x)\coloneqq \frac{\delta^c}{2} x^{-c}$. Then, there exists an even polynomial $P_{c,\delta, \varepsilon}\in \mathbb{R}[x]$, such that 
    \begin{equation*}
        \begin{aligned}
            \abs{P_{c,\delta,\varepsilon}(x) - g(x)}&\le \varepsilon, \forall x\in \interval{\delta}{1}, \\
               \abs{P_{c,\delta,\varepsilon}(x)}&\le 1, \forall x\in \interval{-1}{1}.
        \end{aligned}
    \end{equation*}
    Moreover, the degree of $P_{c,\delta,\varepsilon}(x)$ is $O((c+1)\log (1/\varepsilon)/\delta)$.
\end{lemma}

\begin{lemma}[Polynomial Approximations of Positive Power Functions,~{\cite[Lemma 6]{wang2024quantum}}]
\label{lem:poly-approx-positive-powers}
    For any $c > 0$, $\beta \in (0,1)$, $\nu \in (0,\beta)$, and $\eta \in \left(0,\tfrac{1}{2}\right)$, let
$ g(x) = 2^{-c-1}\beta^{-c} x^{c}$.
There exists an efficiently computable even polynomial $S_{c,\nu,\beta,\eta} \in \mathbb{R}[x]$ of degree $
O\!\left(\frac{c+1}{\nu}\log\!\left(\frac{1}{\beta \nu \eta}\right)\right) $,
such that
\begin{equation*}
    \begin{aligned}
\forall x \in [0,\nu] &:\quad |S(x)| \le 2g(x), \\
\forall x \in [\nu,\beta] &:\quad |g(x) - S(x)| \le \eta, \\
\forall x \in [-1,1] &:\quad |S(x)| \le 1 .
\end{aligned}
\end{equation*}
\end{lemma}

\section{Multi-Level Amplitude Estimation}

In this section, we describe our multi-level estimation framework for the functional $\sum_{i=1}^n p_i g(p_i)$ of a discrete distribution $p$, given quantum query access via a unitary oracle $\mathcal{O}_p$.

We perform the estimation by manipulating a projected unitary encoding of a matrix $A$, which encodes the distribution probabilities in its singular values. Following the construction in \cite{gilyen2020distributional}, we aim to construct $A$ with the singular value decomposition:
\begin{align*} 
    A = \sum_{i=1}^n \frac{\sqrt{p_i}}{2}\ketbra{\tilde{\psi}_i}{\psi_i} = \sum_{i=1}^n \frac{\sqrt{p_i}}{2}\ketbra{i}{i}\otimes \ketbra{\mathbf{0}}{i} \otimes \ketbra{\mathbf{0}}{\phi_i}\otimes  \ketbra{0}{0}.
\end{align*}
To implement this encoding, we use a system of four registers. Let $R$ be a single-qubit unitary that maps $\ket{0} \mapsto \frac{1}{2} \ket{0} + \frac{\sqrt{3}}{2} \ket{1}$. This projected unitary encoding is realized by the tuple $(\widetilde{\Pi}, \Pi,U)$ defined as
\begin{align} \label{eq:def-be}
    \widetilde{\Pi} = I\otimes \ketbra{\mathbf{0}}{\mathbf{0}} \otimes \ketbra{\mathbf{0}}{\mathbf{0}} \otimes \ketbra{0}{0}, \quad
    U =  I \otimes \mathcal{O}_p^{\dagger} \otimes R^{\dagger},  \quad
    \Pi = \sum_{i=1}^n \ketbra{i}{i}\otimes \ketbra{i}{i} \otimes I \otimes \ketbra{0}{0}.
\end{align}
Let $\mathrm{CNOT}$ denote a bit-wise CNOT gate satisfying $\mathrm{CNOT}\colon \ket{\mathbf{0}}\ket{i}\mapsto \ket{i}\ket{i}$. We prepare the initial state
\begin{align} \label{eq:def-U-I}
     \sum_{i=1}^n \sqrt{p_i} \ket{\psi_i} = \sum_{i=1}^n \sqrt{p_i} \ket{i}\ket{i}\ket{\phi_i}\ket{0}
\end{align}
by applying $\mathcal{O}_p$ to the second and third registers, followed by a bit-wise CNOT gate on the first register controlled by the second. 

Let $\rho > 1$ be a constant. We partition the domain $[0,1]$ into $m+1$ intervals:
\begin{align*}
    \Delta_{m+1}=[0,\varphi_m], \quad \Delta_{j}=(\varphi_j, \varphi_{j-1}] \text{ for } j=1, \ldots, m,
\end{align*} 
where $\varphi_j = \rho^{-j}$ for $j = 0, \ldots, m$.

\subsection{Non-destructive singular value discrimination} Let $U$ be a projected unitary encoding of a matrix $A$ with SVD $A = \sum_{i=1}^d \sigma_i \ketbra{\tilde{\psi}_i}{\psi_i}$. A standard method for singular value discrimination used in \cite{wang2024quantum,shin2025near} utilizes QSVT with a polynomial $P$ that approximates a threshold function. Specifically, to distinguish whether a singular value $\sigma_i$ lies below $\gamma/\rho$ or above $\gamma$ for a threshold $\gamma\in(0,1)$, consider an even polynomial $P$ satisfying $P(x) \approx 0$ on $[0, \gamma/\rho]$ and $P(x) \approx 1$ on $[\gamma,1]$. Applying the associated QSVT unitary to the singular vector $\ket{\psi_i}$ produces the state
\begin{align*} 
P(\sigma_i) \ket{\psi_i} + \sqrt{1-P^2(\sigma_i)}\ket{\psi^{\perp}_i},\nonumber
\end{align*}
where $\ket{\psi^{\perp}_i}$ denotes a garbage state orthogonal to the subspace of interest. To partition $\sigma_i$ into finer intervals, multiple discrimination steps are necessary. Note that if $\sigma_i < \gamma$, the produced state contains a non-negligible component of the garbage state $\ket{\psi_i^{\perp}}$. The main drawback is that this garbage state cannot be reused for subsequent tests, since QSVT leaves the orthogonal subspace undefined. This necessitates preparing a fresh copy of $\ket{\psi_i}$ in an additional ancilla register for every new discrimination step, incurring a substantial space overhead. To overcome this and discriminate singular values using constant ancilla qubits, we apply the gapped phase estimation with branch marking technique proposed in~\cite{low2024quantum}. Since the algorithm only applies to Hermitian matrices, we convert the singular value problem into an eigenvalue problem using a Hermitian dilation. Consider the unitary $U_H = \ketbra{0}{1} \otimes U + \ketbra{1}{0} \otimes U^{\dagger}$ acting on an extended system. With respect to the projector $\Pi_H = \ketbra{0}{0} \otimes \widetilde{\Pi} + \ketbra{1}{1} \otimes \Pi$, the operator $U_H$ encodes the Hermitian matrix
    \begin{align*}
        H = \begin{bmatrix}
            0 & A\\
            A^{\dagger} & 0
        \end{bmatrix},
    \end{align*}
    satisfying $\Pi_H U_H \Pi_H = H$. The spectral decomposition of $H$ takes the form $H = \sum_{i, \nu} \nu \sigma_i \ketbra{\phi_{i,\nu}}{\phi_{i,\nu}}$, where the eigenvectors are given by $\ket{\phi_{i,\nu}} = \frac{1}{\sqrt{2}}(\ket{0}\ket{\psi_i} + \nu\ket{1}\ket{\tilde{\psi}_i})$ for $\nu \in \{+1, -1\}$.
\begin{theorem}
\label{thm:nd-svd}
    Let $(\widetilde{\Pi}, \Pi, U)$ be a projected unitary encoding of $A$ with spectral norm $\|A\|\le 1/2$. Let $A = \sum_{i=1}^d \sigma_i \ketbra{\tilde{\psi}_i}{\psi_i}$ be the singular value decomposition of $A$. Let $\ket{\phi_{i,\nu}}$ be the eigenvector of the Hermitian dilation $H$ of $A$. For any precision $\eps_1, \eps_2 \in (0,1]$, threshold $\gamma\in (0, 1/2)$, and constant gap ratio $\rho>1$, there exists a unitary $D$ that maps \begin{align*}  \ket{++} \ket{0}\ket{\phi_{i,\nu}}\mapsto \ket{++}\ket{\xi_{i,\nu}}\ket{\phi_{i,\nu}}
    \end{align*}
    with precision $\eps_1$  for all $i\in [d]$ and $\nu\in \{-1,1\}$, where $\ket{\xi_{i,\nu}}$ satisfies \begin{align*}
        \begin{cases}
        \|\ket{\xi_{i,\nu}}-\nu\ket{0}\|\le \eps_2, &  \text{ if }\sigma_i\in [\gamma, 1),\\
         \|\ket{\xi_{i,\nu}}-\i \ket{1}\|\le \eps_2 , & \text{ if }\sigma_i \in [0, \frac{\gamma}{\rho}]. \\
        \end{cases} 
    \end{align*}
    Implementing $D$ requires $O(\log(1/\eps_1)+\log(1/\eps_2)/\gamma)$ queries to controlled $U$ and $U^{\dagger}$.
\end{theorem}
\begin{proof}
    Since $U_H$ is Hermitian ($U_H^{\dagger} = U_H$), by the Hermitian qubitization (see Corollary~19 of \cite{low2024quantum}),  within the invariant subspace spanned by the image of $\Pi_H$ and $U_H\Pi_H$, the associated quantum walk operator $W = (2\Pi_H-I) U_H$ admits spectral decomposition 
    \begin{align*}
        W = \sum_{i,\nu} \big(e^{\i \theta_{i,\nu}} \ketbra{\phi_{i,\nu}^{+}}{\phi_{i,\nu}^{+}}+ e^{-\i \theta_{i,\nu}}\ketbra{\phi_{i,\nu}^{-}}{\phi_{i,\nu}^{-}}\big)
    \end{align*}
    where $\cos(\theta_{i,\nu}) = \lambda_{i,\nu} = \nu \sigma_i$ and $\ket{\phi_{i,\nu}} = \frac{1}{\sqrt{2}}(\ket{\phi_{i,\nu}^{+}}+\ket{\phi_{i,\nu}^{-}})$. Each eigenvector $\ket{\phi_{i,\nu}}$ of $H$ is associated with two eigenvectors of $W$, which is called the two branches of $\ket{\phi_{i,\nu}}$. According to~\cite[Proposition 22]{low2024quantum}, there exists a unitary $V_{\mathrm{BM}}$ that implements the branch marking operation. It maps\begin{align*}\ket{+}\ket{+}\ket{\phi_{i,\nu}^{\pm}} \mapsto \ket{+}\ket{\pm}\ket{\phi_{i,\nu}^{\pm}} , \quad \ket{+}\ket{-}\ket{\phi_{i,\nu}^{\pm}} \mapsto \ket{+}\ket{\mp}\ket{\phi_{i,\nu}^{\pm}}\end{align*}with precision $\eps_1/2$, for any $i\in[d]$ and $\nu\in\{+1, -1\}$, using $O(\log(1/\eps_1))$ queries to $U_H$.
    By \cite[Proposition 23]{low2024quantum}, there exists a unitary $V_{\mathrm{GPE}}$ that performs the gapped phase estimation for branch marked eigenvectors using $O(\log(1/\eps_2)/\gamma)$ queries to $U_H$: \begin{align*}
        V_{\mathrm{GPE}}\colon \ket{\pm}\ket{0}\ket{\phi_{i,\nu}^{\pm}}\mapsto\ket{\pm}\ket{\xi_{i,\nu}}\ket{\phi_{i,\nu}^{\pm}}, \quad \begin{cases}
        \|\ket{\xi_{i,\nu}}-\nu\ket{0}\|\le \eps_2, &  \text{ if }\sigma_i\in [\gamma, 1),\\
         \|\ket{\xi_{i,\nu}}-\i \ket{1}\|\le \eps_2 , & \text{ if }\sigma_i \in [0, \frac{\gamma}{\rho}]. \\
        \end{cases} 
    \end{align*}
     Applying $V_{\mathrm{BM}}$ to registers $(\mathsf{a},\mathsf{b},\mathsf{d})$ of $\ket{+}_{\mathsf{a}}\ket{+}_{\mathsf{b}}\ket{0}_{\mathsf{c}}\ket{\phi_{i,\nu}}_{\mathsf{d}}$ yields an $\eps_1/2$-approximation of \begin{align*}
        \ket{+}_{\mathsf{a}}\frac{1}{\sqrt{2}}(\ket{+}_{\mathsf{b}}\ket{0}_{\mathsf{c}}\ket{\phi_{i,\nu}^{+}}_{\mathsf{d}}+\ket{-}_{\mathsf{b}}\ket{0}_{\mathsf{c}}\ket{\phi_{i,\nu}^{-}}_{\mathsf{d}}).
    \end{align*} Then, applying $V_{\mathrm{GPE}}$ to registers $(\mathsf{b},\mathsf{c},\mathsf{d})$ yields an $\eps_1/2$-approximation of \begin{align*}
        \ket{+}_{\mathsf{a}}\frac{1}{\sqrt{2}}(\ket{+}_{\mathsf{b}}\ket{\xi_{i,\nu}}_{\mathsf{c}}\ket{\phi_{i,\nu}^{+}}_{\mathsf{d}}+\ket{-}_{\mathsf{b}}\ket{\xi_{i,\nu}}_{\mathsf{c}}\ket{\phi_{i,\nu}^{-}}_{\mathsf{d}}).
\end{align*}Finally, we apply $V_{\mathrm{BM}}^{\dagger}$ to registers $(\mathsf{a},\mathsf{b},\mathsf{d})$, which uncomputes the branch mark, yielding an $\eps_1$-approximation of  \begin{align*}
        \ket{+}_{\mathsf{a}}\frac{1}{\sqrt{2}}(\ket{+}_{\mathsf{b}}\ket{\xi_{i,\nu}}_{\mathsf{c}}\ket{\phi_{i,\nu}^{+}}_{\mathsf{d}}+\ket{+}_{\mathsf{b}}\ket{\xi_{i,\nu}}_{\mathsf{c}}\ket{\phi_{i,\nu}^{-}}_{\mathsf{d}}) = \ket{+}_{\mathsf{a}}\ket{+}_{\mathsf{b}}\ket{\xi_{i,\nu}}_{\mathsf{c}}\ket{\phi_{i,\nu}}_{\mathsf{d}}.
    \end{align*} The whole unitary $(I\otimes V_{\mathrm{BM}}^{\dagger})V_{\mathrm{GPE}}(I\otimes V_{\mathrm{BM}})$ implements the map $\ket{++}\ket{0}\ket{\phi_{i,\nu}}\mapsto \ket{++}\ket{\xi_{i,\nu}}\ket{\phi_{i,\nu}}$  using \begin{align*}
        O\left(\log\left(\frac{1}{\eps_1}\right)+\frac{1}{\gamma}\log\left(\frac{1}{\eps_2}\right)\right) 
    \end{align*}
    queries to $U_H$, and each query to $U_H$ uses one query to controlled $U$ and $U^{\dagger}$. This completes the proof. 
\end{proof}

\subsection{Main algorithm and its analysis}
\label{sec:main-algo}
The pseudocode of the algorithm is presented in \algo{multi-level-ae}.  For each $j\in[m]$, we isolate singular values in the $j$-th loop via two non-destructive singular value discriminators $D_{j-1}$ and $D_{j}$. 

In the following theorem, we show that this process effectively partitions the $i$-th probability of the distribution into $m$ groups.  Specifically, each $\sqrt{p_i}$ is assigned a probability weight $|\beta_{i,j}|^2$ for the $j$-th group, satisfying two key properties. Suppose $\sqrt{p_i}\in \Delta_{j_i}$, and then we have 
\begin{enumerate}
    \item \emph{(Localization of $\beta_{i,j}$)} $|\beta_{i,j}|^2\approx 0$ for any $j \notin \{j_i-1, j_i\}$,
    \item \emph{(Completeness of $\beta_{i,j}$)} $|\beta_{i,j_i-1}|^2+|\beta_{i,j_i}|^2\approx 1$.
\end{enumerate} 
We then apply a QSVT with polynomial $P_j$ to the isolated singular values, followed by amplitude estimation to estimate the squared norm of the output state. Taking the sum over all groups yields an estimate of:
\begin{align}
\label{eq:estimate-sum}
    \sum_{j=1}^m\sum_{i=1}^{n} p_i |\beta_{i,j}|^2 P_j^2\left(\frac{\sqrt{p_i}}{2}\right).
\end{align}
Due to the localization of $\beta_{i,j}$, the polynomial $P_j^2(\sqrt{p_i}/2)$ only needs to approximate $g(p_i)$ for $\sqrt{p_i} \in \Delta_j\cup \Delta_{j+1}$. The completeness property then ensures that \eq{estimate-sum} approximates the target functional $\sum_{i=1}^n p_i g(p_i)$.

\begin{algorithm}[htb]
\caption{Multi-Level Estimation of Functionals}
\label{algo:multi-level-ae}
\Input{
    Unitary oracle $\mathcal{O}_p$ for distribution $(p_i)_{i\in[n]}$; target function $g$; precision $\eps$\;
    Number of intervals $m$, decreasing ratio $\rho$, and upper bounds $\{B_j\}_{j=0}^{m+2}$\;
    Polynomials $\{P_j\}_{j=1}^m$ with the same definite parity\;
}
\Registers{
    Ancilla qubits $\mathsf{a}, \mathsf{b}, \mathsf{c}, \mathsf{d}$, and system $\mathsf{s} = (\mathsf{s}_1,\mathsf{s}_2)$\;
}
\Output{
    Estimate of $\sum_{i=1}^n p_i g(p_i)$.
}

\BlankLine
$S \leftarrow 0$\;
Let $(\widetilde{\Pi}, \Pi, U)$ be the projected unitary encoding of $\sum_{i=1}^n {\sqrt{p_i}}/{2}\ketbra{\tilde{\psi}_i}{\psi_i}$ defined in \eq{def-be} and $(\Pi_H, \Pi_H, U_H)$ be the projected unitary encoding of its Hermitian dilation\;
Let $U_I$ be the state preparation unitary defined in \eq{def-U-I}\;
For any $j\in [m]$, let $D_j$ be the singular value discriminator $D$ in \thm{nd-svd} with threshold $\rho^{-j}/2$, gap ratio $\rho$, and precision $\eps_1 = \min\{1, \eps/(24m\max\{B_j, B_{j+1}, B_{j+2}\})\}$, $\eps_2 =  \min\{1, \eps/(6m\max\{B_j, B_{j+1}, B_{j+2}\})\}$; Let $D_0 = I$\;

\For{$j \leftarrow 1$ \KwTo $m$}{
    Let $W_j$ be the QSVT unitary $U_{P_j}^{(\mathrm{SV})}$ in \thm{qsvt} for the polynomial $P_j$ and the projected unitary encoding $(\Pi_H, \Pi_H, U_H)$ \label{lin:W-j}\;
    
    Let $V_j$ be the sequence: apply $D_{j-1}$ on $(\mathsf{a}, \mathsf{b},\mathsf{c},  \mathsf{s})$, then $D_j$ on $(\mathsf{a},  \mathsf{b}, \mathsf{d},\mathsf{s})$, and finally $W_j$ on $(\mathsf{a}, \mathsf{s})$\;
    
    Let $\ket{\Psi_{\mathrm{in}}} = \ket{00}_{\mathsf{c}, \mathsf{d}}\ket{++}_{\mathsf{a}, \mathsf{b}} \otimes \ket{0}_{\mathsf{s}_1}(U_I \ket{\mathbf{0}}_{\mathsf{s}_2})$ be the initial state\;
    Let  $\Pi_{j} = \begin{cases} 
    \ketbra{+}{+}_{\mathsf{a}} \otimes I_{\mathsf{b},\mathsf{c}} \otimes\ketbra{0}{0}_{\mathsf{d}}\otimes (\Pi_H)_\mathsf{s} & \text{if } j=1 \\
     \ketbra{+}{+}_{\mathsf{a}} \otimes I_{\mathsf{b}}\otimes \ketbra{10}{10}_{\mathsf{c}, \mathsf{d}} \otimes (\Pi_H)_\mathsf{s} & \text{if } j\ge 2
\end{cases}$\;
    
    Estimate the squared norm $\|\Pi_{j} V_j \ket{\Psi_{\mathrm{in}}} \|^2 $ to additive error $\min\{1, \eps/(6m\max\{B_j, B_{j+1}\})\}$ via amplitude estimation (\cor{two-stage-ae}) with $\eta = 1/(3m)$, denoting the output by $\tilde{v}_j$ \label{lin:v-j}\;
    
    $S \leftarrow S + \max\{B_j, B_{j+1}\}  \tilde{v}_j$\;
}
\Return{$S$}\;
\end{algorithm}

\begin{theorem}\label{thm:multi-level-ae-1}
    Let $\eps \in (0,1)$ be a precision parameter, $m$ be the number of intervals, and $g\colon[0,1]\to\mathbb{R}$ be the target function. 
    Given a constant $\rho >1$, define the interval endpoints $\varphi_j = \rho^{-j}$ for $j=0, \ldots, m+1$. 
    Let $B_1, \ldots, B_{m+1}\ge 0$ be upper bounds satisfying $B_j \ge 2\max_{x\in (\varphi_j ,\varphi_{j-1}]} |g(x^2)|$, with convention $B_0=B_{m+2}= 0$. 
    For each $j \in [m]$, define the local maximum bound $\mathcal{B}_j := \max\{B_{j-1}, B_j, B_{j+1}, B_{j+2}\}$.
    Suppose there exist a constant $C > 0$ and polynomials $P_1, \ldots, P_m\in \mathbb{R}[x]$ of the same definite parity satisfying the following four conditions.

   \begin{enumerate}[label=\textbf{Condition \arabic*.}, ref=\arabic*, leftmargin=*]
    \item \label{cond:bound-P-j-1} For all $j \in [m]$ and $x\in[-1,1]$, $|P_j(x)|\le 1$.
    
    \item \label{cond:approx-P-j} For all $j \in [m]$ and $x \in (\varphi_{j+1}, \varphi_{j-1}]$, $\left|\max\{B_j, B_{j+1}\}P_j^2(x/2)-g(x^2)\right| \le \varepsilon/12m$.

    \item \label{cond:bound-P-j} For any $x\in (\varphi_{m+1},\varphi_m]$, $ B_m P_m^2(x/2) \le C|g(x^2)|$.

    \item \label{cond:tail-bound} It holds that
\end{enumerate}
\vspace{-1.3mm}
    \[
        \sum_{i\in [n]: \sqrt{p_i}\in [0,\varphi_m]} p_i|g(p_i)|\le \frac{\eps}{3(1+C)}.
    \]
    Furthermore, assume the degrees satisfy $\deg(P_j) = O(\alpha_j(\eps, m)\log(m\mathcal{B}_j/\eps)/\varphi_{j+1})$ for some functions $\alpha_j(\eps,m)\ge 1$. Given purified quantum query access to any probability distribution $(p_i)_{i\in[n]}$, there exists a quantum algorithm estimating $\sum_{i=1}^n p_i g(p_i)$ to within additive error $\eps$ with probability at least $2/3$ using 
    \begin{align*}
        O\bigg( \sum_{j=1}^m\log(m)\alpha_j(\eps,m)\log\left(\frac{m\mathcal{B}_j}{\eps}\right)\bigg(\frac{m\mathcal{B}_j\sqrt{n_j}}{\eps} + \frac{1}{\varphi_{j+1}}\sqrt{\frac{m\mathcal{B}_j}{\eps}}\bigg) \bigg)
    \end{align*} 
    queries, where $n_j = |\{i\in [n]:\sqrt{p_i}\in (\varphi_{j+1}, \varphi_{j-1}]\}|$ denotes the number of indices in the local neighborhood of the $j$-th interval. The algorithm uses four additional qubits beyond the register required for the projected unitary encoding of the distribution $p$. 
\end{theorem}
A few remarks on the technical conditions are in order to clarify their meaning:
\begin{itemize}[leftmargin=*]
    \item \textbf{\cond{bound-P-j-1} (QSVT validity):} This ensures that the polynomials are bounded and valid for implementation via QSVT.
    \item \textbf{\cond{approx-P-j} (Local approximation):} This guarantees that the scaled polynomial approximates the target function with high accuracy within the $j$-th interval and its transition boundaries.
    \item \textbf{\cond{bound-P-j} and \cond{tail-bound} (Well-bounded tail):} These conditions handle the residual interval $[0, \varphi_m]$ closest to the singularity at the origin. In this region, we do not require an accurate polynomial approximation. Instead, \cond{tail-bound} ensures the true functional contribution from these  small probabilities is sufficiently small, while \cond{bound-P-j} restricts the magnitude of the polynomial estimate. Together, they guarantee via the triangle inequality that the total estimation error from this tail region is well-bounded.
\end{itemize}
\begin{proof}[Proof of \thm{multi-level-ae-1}]
    We use \algo{multi-level-ae} to estimate $\sum_{i=1}^n p_i g(p_i)$.  
    By \thm{nd-svd}, for any $j\in [m]$ and $i\in [n]$, the discriminator $D_j$ maps \begin{align*}
        \ket{++}\ket{0}\ket{0}\ket{\psi_i} = \ket{++}\ket{0}\frac{1}{\sqrt{2}}(\ket{\phi_{i,+1}}+\ket{\phi_{i,-1}})\mapsto\ket{++}\frac{1}{\sqrt{2}}(\ket{\xi_{i,+1}}\ket{\phi_{i,+1}}+\ket{\xi_{i,-1}}\ket{\phi_{i,-1}}),
    \end{align*} with accuracy $\eps/(24m\max\{B_j, B_{j+1}, B_{j+2}\})$, where \begin{align}
        \|\ket{\xi_{i,j,\pm 1}}-\i \ket{1}\|&\le \min\left\{1, \frac{\eps}{6m\max\{B_j, B_{j+1}, B_{j+2}\}}\right\} , && \text{ for }\sqrt{p_i} \le \varphi_{j+1}\text{,} \label{eq:cond-xi-s}\\
        \|\ket{\xi_{i,j,\pm 1}}-\pm\ket{0}\|&\le \min\left\{1, \frac{\eps}{6m\max\{B_j, B_{j+1}, B_{j+2}\}}\right\}, &&  \text{ for }\sqrt{p_i} \ge \varphi_{j}\label{eq:cond-xi-l}.
    \end{align} For consistency, we define $\ket{\xi_{i,0,\pm 1}}=\pm \ket{0}$ since $D_0 = I$. Define the target flag state for interval $j$ as $$\ket{F_j} = \begin{cases} \ket{00} & \text{if } j=1, \\ \ket{10} & \text{if } j \ge 2. \end{cases}$$ Then, for any $j\in[m]$ and $i\in[n]$, applying $D_{j-1}$ to registers $\mathsf{a}, \mathsf{b},\mathsf{c}, \mathsf{s}$ and $D_j$ to registers $\mathsf{a},\mathsf{b}, \mathsf{d}, \mathsf{s}$ implements the map \begin{align}
    \ket{++}_{\mathsf{a},\mathsf{b}}\ket{0}_{\mathsf{c}} \ket{0}_{\mathsf{d}}\ket{0}_{\mathsf{s}_1}\ket{\psi_i}_{\mathsf{s}_2}\mapsto &\ket{++}_{\mathsf{a},\mathsf{b}}\frac{1}{\sqrt{2}}(\ket{\xi_{i,j-1,+1}}_{\mathsf{c}} \ket{\xi_{i,j,+1}}_{\mathsf{d}}\ket{\phi_{i,+1}}_{\mathsf{s}}+ \ket{\xi_{i,j,-1}}_{\mathsf{c}} \ket{\xi_{i,j,-1}}_{\mathsf{d}}\ket{\phi_{i,-1}}_{\mathsf{s}})\nonumber\\
    =:&\ket{++}_{\mathsf{a},\mathsf{b}}\frac{1}{\sqrt{2}}(\beta_{i,j,+1}\ket{F_j}_{\mathsf{c},\mathsf{d}}\ket{\phi_{i,+1}}_{\mathsf{s}}+\beta_{i,j,-1}\ket{F_j}_{\mathsf{c},\mathsf{d}}\ket{\phi_{i,-1}}_{\mathsf{s}}+\ket{\perp}_{\mathsf{c}, \mathsf{d}, \mathsf{s}})\label{eq:state-separate}
    \end{align}
    with accuracy $\eps/(12m\max\{B_j, B_{j+1}\})$, where $\ket{\perp}_{\mathsf{c}, \mathsf{d}, \mathsf{s}}$ is the component orthogonal to $\ket{0}_{\mathsf{d}}$ for $j=1$ and orthogonal to $\ket{10}_{\mathsf{c,d}}$ for $j \ge 2$. We define $\beta_{i,0,\pm1}=0$. The amplitude $\beta_{i,j,\pm1}$ satisfies \begin{align*}
        |\beta_{i,j,\pm1}|^2 = \begin{cases}
        |\braket{0}{\xi_{i,j-1, \pm1}}|^2 |\braket{0}{\xi_{i,j, \pm1}}|^2  & \text{if }j=1, \\
        |\braket{1}{\xi_{i,j-1, \pm1}}|^2 |\braket{0}{\xi_{i,j, \pm1}}|^2 & \text{if } j \ge 2.
        \end{cases}
    \end{align*} 
    For any $j\in[m]$, the unitary $W_j$ in \lin{W-j} maps \begin{align}
    \label{eq:eff-W-j}
        \ket{+} \ket{\phi_{i,\pm 1}} \mapsto P_j\left(\pm\frac{\sqrt{p_i}}{2}\right) \ket{+} \ket{\phi_{i,\pm 1}} +  \ket{\gamma_{i,j,\pm1}},
    \end{align}
    where $\ket{\gamma_{i,j,\pm1}}$ satisfies $(\ketbra{+}{+}\otimes \Pi_H)\ket{\gamma_{i,j,\pm1}}=0$. 
    Combining \eq{state-separate} and \eq{eff-W-j}, the final state in the $j$-th loop $\ket{\Psi_j}:=V_j\ket{\Psi_{\mathrm{in}}}$ is $\eps/(12m\max\{B_j, B_{j+1}\})$-close to   
    \begin{align*}
        \sum_{i=1}^n \left(\sqrt{\frac{p_i}{2}}\ket{++}_{\mathsf{a},\mathsf{b}}\ket{F_j}_{\mathsf{c},\mathsf{d}}\left(\beta_{i,j,+1}P_j\left(\frac{\sqrt{p_i}}{2}\right)\ket{\phi_{i,+1}}_{\mathsf{s}}+\beta_{i,j,-1}P_j\left(-\frac{\sqrt{p_i}}{2}\right)\ket{\phi_{i,-1}}_{\mathsf{s}}\right)+\ket{\eta_{i,j}}_{\mathsf{a},\mathsf{b},\mathsf{c},\mathsf{d},\mathsf{s}}\right).
    \end{align*}
    Here, $\ket{\eta_{i,j}}$ is an unnormalized state satisfying $\Pi_j\ket{\eta_{i,j}}=0$ with \begin{align*}
        \Pi_{j} := \begin{cases} 
    \ketbra{+}{+}_{\mathsf{a}} \otimes I_{\mathsf{b},\mathsf{c}} \otimes\ketbra{0}{0}_{\mathsf{d}}\otimes (\Pi_H)_\mathsf{s} & \text{if } j=1 \\
     \ketbra{+}{+}_{\mathsf{a}} \otimes I_{\mathsf{b}}\otimes \ketbra{10}{10}_{\mathsf{c}, \mathsf{d}} \otimes  (\Pi_H)_\mathsf{s} & \text{if } j\ge 2.
     \end{cases}
    \end{align*} Define the effective amplitude $\beta_{i,j}$ via its squared magnitude \begin{align*}
        |\beta_{i,j}|^2:= \frac{1}{2}(|\beta_{i,j,+1}|^2+ |\beta_{i,j,-1}|^2).
    \end{align*}Then, the squared norm $\|\Pi_j V_j \ket{\Psi_{\mathrm{in}}}\|^2$ is $\eps/(6m\max\{B_j, B_{j+1}\})$-approximate to \begin{align*}
       \sum_{i=1}^{n} p_i \frac{1}{2}\left(|\beta_{i,j,+1}|^2 P_j^2\left(\frac{\sqrt{p_i}}{2}\right) +|\beta_{i,j,-1}|^2 P_j^2\left(-\frac{\sqrt{p_i}}{2}\right)\right) & = \sum_{i=1}^{n} p_i |\beta_{i,j}|^2 P_j^2\left(\frac{\sqrt{p_i}}{2}\right),
    \end{align*} since $P_j(x)$ has definite parity. By \cor{two-stage-ae}, the output $\tilde{v}_j$ in \lin{v-j} is an estimate of $\|\Pi_j V_j \ket{\Psi_{\mathrm{in}}}\|^2$ with error at most $\eps/(6m\max\{B_j, B_{j+1}\})$ and success probability at least $1-1/3m$. By the union bound, the output $S$ of \algo{multi-level-ae} satisfies \begin{align}
    \label{eq:ae-error}
        \bigg| S-\sum_{j=1}^m \sum_{i=1}^{n}\max\{B_j, B_{j+1}\} p_i |\beta_{i,j}|^2 P_j^2\left(\frac{\sqrt{p_i}}{2}\right)\bigg| \le  \sum_{j=1}^m \frac{\eps}{3m} \le \frac{\eps}{3}
    \end{align} 
    with probability at least $2/3$.
    Then it suffices to show that $\sum_{j=1}^m \sum_{i=1}^{n}\max\{B_j, B_{j+1}\} p_i |\beta_{i,j}|^2 P_j^2(\sqrt{p_i}/2)$ is a $2\eps/3$-approximation of $\sum_{i=1}^n p_i g(p_i)$. We first prove the two properties of $\beta_{i,j}$.
    
    \paragraph{Localization of $\beta_{i,j}$.} For a fixed $j$ with $2 \le j \le m$, and any $\sqrt{p_i}$ in $[0,\varphi_{j+1}]$ and $(\varphi_{j-1}, 1]$, \eq{cond-xi-s} and \eq{cond-xi-l} imply  \begin{align*}
        \|\ket{\xi_{i,j,\pm1}}-\i \ket{1}\|&\le \min\left\{1, \frac{\eps}{3m\max\{B_j,B_{j+1}\}}\right\}\quad \text{and}\\
        \|\ket{\xi_{i,j-1,\pm1}}-\pm\ket{0}\|&\le \min\left\{1, \frac{\eps}{3m\max\{B_j,B_{j+1}\}}\right\}
    \end{align*} respectively. Consequently, the amplitude $\beta_{i,j}$ satisfies 
    \begin{align}
        \label{eq:beta-cond-1-j>1}
        |\beta_{i,j}|^2 &= \frac{1}{2}\sum_{\nu\in \{-1,+1\}}\abs*{\braket{1}{\xi_{i,j-1,\nu}}}^2 \cdot \abs*{\braket{0}{\xi_{i,j,\nu}}}^2 \nonumber\\
        &\le \left(\min\left\{1, \frac{\eps}{3m\max\{B_j,B_{j+1}\}}\right\}\right)^4 \le\frac{\eps}{3m\max \{B_j,B_{j+1}\}},\quad \text{for } 2 \le j \le m.
    \end{align} For $j = 1$,  any $\sqrt{p_i} \in[0,\varphi_{j+1}]\cup (\varphi_{j-1}, 1] = [0,\varphi_{2}]$, by \eq{cond-xi-s} and the definition of $\ket{\xi_{i,0,\pm1}}$, we have \begin{align*}
        \|\ket{\xi_{i,j,\pm1}}-\i \ket{1}\|\le \min\left\{1, \frac{\eps}{3m\max\{B_j,B_{j+1}\}}\right\},\quad \ket{\xi_{i,j-1,\pm1}} = \pm\ket{0}. 
    \end{align*} Consequently, the amplitude $\beta_{i,j}$ satisfies \begin{align}
        \label{eq:beta-cond-1-j=1}
        |\beta_{i,j}|^2 &= \frac{1}{2}\sum_{\nu\in \{-1,+1\}}\abs*{\braket{0}{\xi_{i,j-1,\nu}}}^2 \cdot \abs{\braket{0}{\xi_{i,j,\nu}}}^2 \nonumber\\
        &\le \left(\min\left\{1, \frac{\eps}{3m\max\{B_j,B_{j+1}\}}\right\}\right)^2 \le \frac{\eps}{3m\max\{B_j,B_{j+1}\}}, \quad \text{for } j = 1.
    \end{align} In summary, combining \eq{beta-cond-1-j>1} and \eq{beta-cond-1-j=1}, we conclude that for any $j \in [m]$ and $\sqrt{p_i} \in [0,\varphi_{j+1}] \cup (\varphi_{j-1}, 1]$, the amplitude $\beta_{i,j}$ is bounded by
    \begin{align*}
        |\beta_{i,j}|^2 \le \frac{\eps}{3m\max \{B_j,B_{j+1}\}},
    \end{align*}
    which we call \emph{the localization bound of $\beta_{i,j}$}.
    \paragraph{Completeness of $\beta_{i,j}$.}
    For a fixed $j$ with $3 \le j \le m$ and any $\sqrt{p_i} \in [\varphi_{j}, \varphi_{j-1}]$, \eq{cond-xi-s} and \eq{cond-xi-l} imply  $\|\ket{\xi_{i,j-2,\pm 1}}-\i \ket{1}\|\le \eps/(6m B_j)$ and $\|\ket{\xi_{i,j,\pm 1}}-\pm\ket{0}\|\le \eps/(6mB_j)$. Then the amplitude $\beta_{i,j-1}$ and $\beta_{i,j}$ satisfies 
    \begin{align}
        |\beta_{i,j-1}|^2 + |\beta_{i,j}|^2 &=\frac{1}{2}\sum_{\nu\in \{-1,+1\}}\big(\abs*{\braket{1}{\xi_{i,j-2,\nu}}}^2 \cdot \abs*{\braket{0}{\xi_{i,j-1,\nu}}}^2 + \abs*{\braket{1}{\xi_{i,j-1,\nu}}}^2\cdot \abs{\braket{0}{\xi_{i,j,\nu}}}^2 \big)\nonumber \\
        &\ge \left(1- \frac{\eps}{6mB_j}\right)^2\frac{1}{2}\sum_{\nu\in \{-1,+1\}} (\abs*{\braket{0}{\xi_{i,j-1,\nu}}}^2+  \abs*{\braket{1}{\xi_{i,j-1,\nu}}}^2) \nonumber \\
        &=\left(1- \frac{\eps}{6mB_j}\right)^2 \ge 1-\frac{\eps}{3mB_j}.\label{eq:beta-cond-2-j>2}
    \end{align}
    For $j = 2$ and any $\sqrt{p_i} \in [\varphi_{j}, \varphi_{j-1}]$, by the definition of $\ket{\xi_{i,0,\pm 1}}$ and \eq{cond-xi-l}, we have $\ket{\xi_{i,j-2,\pm 1}} = \pm \ket{0}$ and $\|\ket{\xi_{i,j,\pm 1}}-\pm \ket{0}\|\le \eps/(6mB_j)$. Then the amplitude $\beta_{i,j-1}$ and $\beta_{i,j}$ satisfy \begin{align}
        |\beta_{i,j-1}|^2 + |\beta_{i,j}|^2 &=\frac{1}{2}\sum_{\nu\in \{-1,+1\}}\big(\abs*{\braket{0}{\xi_{i,j-2,\nu}}}^2 \cdot \abs*{\braket{0}{\xi_{i,j-1,\nu}}}^2 + \abs*{\braket{1}{\xi_{i,j-1,\nu}}}^2 \cdot \abs*{\braket{0}{\xi_{i,j,\nu}}}^2 \big)\nonumber \\
        &\ge \left(1- \frac{\eps}{6mB_j}\right)^2\frac{1}{2}\sum_{\nu\in \{-1,+1\}} (\abs*{\braket{0}{\xi_{i,j-1,\nu}}}^2+  \abs*{\braket{1}{\xi_{i,j-1,\nu}}}^2) \nonumber \\
        &=\left(1- \frac{\eps}{6mB_j}\right)^2 \ge 1-\frac{\eps}{3mB_j}.\label{eq:beta-cond-2-j>1}
    \end{align} 
    For $j = 1$ and any $\sqrt{p_i} \in [\varphi_{j}, \varphi_{j-1}]=[\varphi_{1}, \varphi_{0}]$, we have $\beta_{i,0,\pm1}=0$ and $\ket{\xi_{i,0,\pm 1}}=\pm\ket{0}$ by definition, and $\|\ket{\xi_{i,1,\pm 1}}-\pm \ket{0}\|\le \eps/(6mB_1)$. Consequently, the amplitude $\beta_{i,1}$ satisfies \begin{align}
        |\beta_{i,0}|^2 +|\beta_{i,1}|^2 &=|\beta_{i,1}|^2 \nonumber\\
        &=  \frac{1}{2}\sum_{\nu\in \{-1,+1\}} \abs*{\braket{0}{\xi_{i,0,\nu}}}^2 \cdot \abs*{\braket{0}{\xi_{i,1,\nu}}}^2 \nonumber\\
        &=  \frac{1}{2}\sum_{\nu\in \{-1,+1\}}\abs*{\braket{0}{\xi_{i,1,\nu}}}^2 \ge  \left(1- \frac{\eps}{6mB_1}\right)^2 \ge 1-\frac{\eps}{3mB_1}. \label{eq:beta-cond-2-j=1}
    \end{align}
    In summary, combining \eq{beta-cond-2-j>2}, \eq{beta-cond-2-j>1}, and \eq{beta-cond-2-j=1}, we conclude that for any $j \in [m]$ and $\sqrt{p_i} \in [\varphi_{j}, \varphi_{j-1}]$, the amplitude $\beta_{i,j-1}$ and $\beta_{i,j}$ satisfy 
    \begin{align}
        \label{eq:beta-cond-2}
        |\beta_{i,j-1}|^2 + |\beta_{i,j}|^2 \ge 1-\frac{\eps}{3mB_j}.
    \end{align} For the corresponding upper bound, bounding the outer projection probabilities by $1$ directly yields
    \begin{align*}
        |\beta_{i,j-1}|^2 + |\beta_{i,j}|^2 &\le \frac{1}{2}\sum_{\nu\in \{-1,+1\}}\Big( 1\cdot \abs*{\braket{0}{\xi_{i,j-1,\nu}}}^2 + \abs*{\braket{1}{\xi_{i,j-1,\nu}}}^2 \cdot 1\Big) \le 1.
    \end{align*}
    Combining this with~\eq{beta-cond-2}, we obtain the final completeness condition:
    \begin{align*}
         \left| |\beta_{i,j-1}|^2 + |\beta_{i,j}|^2 - 1 \right| \le \frac{\eps}{3mB_j},
    \end{align*}
    which we call the \emph{completeness of $\beta_{i,j}$}. 
\paragraph{Error analysis.}
    For any $i\in [n]$, let $j_i\in[m+1]$ be the index such that $\sqrt{p_i}\in \Delta_{j_i}$. If $j_i \in [m]$, for any $j\in [m]$ satisfying $j\le j_i-2$ or $j \ge j_i+1$, we have $\sqrt{p_i} \in (\varphi_{j_i}, \varphi_{j_i-1}] \subseteq [0,\varphi_{j+1}]\cup (\varphi_{j-1}, 1]$, and hence \begin{align}
        \label{eq:beta-cond-3}
        |\beta_{i,j}|^2 \le \frac{\eps} {3m \max\{ B_j, B_{j+1}\}} 
    \end{align} by the localization bound of $\beta_{i,j}$.
    Then we have \begin{align}
        &\Big|\max\{B_{j_i-1}, B_{j_i}\}|\beta_{i, j_i-1}|^2 P_{j_i-1}^2\left(\frac{\sqrt{p_i}}{2}\right)+\max\{B_{j_i}, B_{j_i+1}\}|\beta_{i, j_i}|^2 P_{j_i}^2\left(\frac{\sqrt{p_i}}{2}\right)- g(p_i)\Big|\nonumber \\
        \le ~ & \begin{multlined}[t]
            |\beta_{i, j_i-1}|^2 |\max\{B_{j_i-1}, B_{j_i}\}P_{j_i-1}^2\left(\frac{\sqrt{p_i}}{2}\right)-g(p_i)|+|\beta_{i, j_i}|^2 |\max\{B_{j_i}, B_{j_i+1}\}P_{j_i}^2\left(\frac{\sqrt{p_i}}{2}\right)- g(p_i)|\\ +\big||\beta_{i, j_i-1}|^2+|\beta_{i, j_i}|^2 -1\big||g(p_i)| 
        \end{multlined} \nonumber\\
        \le ~&\frac{\eps}{12m}(|\beta_{i, j_i-1}|^2+|\beta_{i, j_i}|^2)+ \frac{\eps}{3mB_{j_i}}|g(p_i)| \nonumber\\
        \le ~&  \frac{\eps}{6m}+ \frac{\eps}{6m}=\frac{\eps}{3m}.\label{eq:concentrate-beta}
    \end{align} 
    The second inequality follows from \cond{approx-P-j} and the completeness of $\beta_{i,j}$. The third inequality follows from the definition of $B_j$.  Then, we have \begin{align}
        &\bigg|\sum_{j=1}^m\max\{B_{j_i}, B_{j_i+1}\}|\beta_{i, j}|^2 P_j^2\left(\frac{\sqrt{p_i}}{2}\right)- g(p_i)\bigg| \nonumber\\
        \le~& \begin{multlined}[t]
            \sum_{j\in[m]:j\neq j_i-1, j_i}\max\{B_{j}, B_{j+1}\}|\beta_{i, j}|^2 P_j^2\left(\frac{\sqrt{p_i}}{2}\right)\\+ \Big|\max\{B_{j_{i}-1}, B_{j_i}\}|\beta_{i, j_i-1}|^2 P_{j_i-1}^2\left(\frac{\sqrt{p_i}}{2}\right)+\max\{B_{j_i}, B_{j_i+1}\}|\beta_{i, j_i}|^2 P_{j_i}^2\left(\frac{\sqrt{p_i}}{2}\right)- g(p_i)\Big| 
        \end{multlined}\nonumber\\
        \le~&\sum_{j\in[m]:j\neq j_i-1, j_i}\frac{\eps}{3m}+ \frac{\eps}{3m}\nonumber\\
        \le ~& \frac{\eps}{3}, \label{eq:approx-f-<=m}
    \end{align}
    where the third line follows from \cond{bound-P-j-1}, \eq{beta-cond-3}, and \eq{concentrate-beta}.
    If $j_i = m+1$, which means $\sqrt{p_i} \in [0, \varphi_m]$, then for any $j\le m-1$, we have $\sqrt{p_i} \in [0, \varphi_{j+1}]$, and hence \begin{align}
    \label{eq:beta-bound-4}
         |\beta_{i,j}|^2 \le\frac{\eps} {3m \max\{ B_j, B_{j+1}\}} .
    \end{align} by the localization bound of $\beta_{i,j}$. Therefore, if $\sqrt{p_i}\in (\varphi_{m+1}, \varphi_m]$, we have 
    \begin{align}
        & \bigg|\sum_{j=1}^m\max\{B_j, B_{j+1}\}|\beta_{i, j}|^2 P_j^2\left(\frac{\sqrt{p_i}}{2}\right)- g(p_i)\bigg| \nonumber \\
        \le {} & \sum_{j=1}^{m-1} \max\{B_j, B_{j+1}\}|\beta_{i, j}|^2 +B_m|\beta_{i, m}|^2P_m^2\left(\frac{\sqrt{p_i}}{2}\right)+|g(p_i)|\nonumber\\
        \le {} & \sum_{j=1}^{m-1}\frac{\eps}{3m}+ C|\beta_{i, m}|^2|g(p_i)|+|g(p_i)|\nonumber\\
        \le {} & \frac{\eps}{3} + (1+C)|g(p_i)|,\label{eq:approx-f-=m+1}
    \end{align}
    where the first line follows from \cond{bound-P-j-1}, and the second line follows from \cond{bound-P-j} and \eq{beta-bound-4}. If $\sqrt{p_i} \in [0, \varphi_{m+1}]$, by the localization bound of $\beta_{i,j}$, we have $|\beta_{i,m}|^2 \le \eps/(3m B_m)$, and hence $B_m|\beta_{i, m}|^2P_m^2(\sqrt{p_i}/2) \le \eps/(3m)$ by \cond{bound-P-j-1}. Then, analogously to \eq{approx-f-=m+1}, we have \begin{align}
         \bigg|\sum_{j=1}^m\max\{B_j, B_{j+1}\}|\beta_{i, j}|^2 P_j^2\left(\frac{\sqrt{p_i}}{2}\right)- g(p_i)\bigg| 
        &\le \sum_{j=1}^m \frac{\eps}{3m}+|g(p_i)|= \frac{\eps}{3} + |g(p_i)|.\label{eq:approx-f-=m+2}
    \end{align} Combining \eq{ae-error}, \eq{approx-f-<=m}, \eq{approx-f-=m+1}, and \eq{approx-f-=m+2}, the output $S$ satisfies 
    \begin{align}
        \Big|S-\sum_{i=1}^m p_i g(p_i)\Big|&\le  \sum_{i=1}^n p_i \Big|\sum_{j=1}^m\max\{B_j, B_{j+1}\}|\beta_{i, j}|^2 P_j^2\left(\frac{\sqrt{p_i}}{2}\right)- g(p_i)\Big|+\frac{\eps}{3}\nonumber\\
        &\le  (1+C)\sum_{i\in [n]: \sqrt{p_i}\in \Delta_{m+1}} p_i|g(p_i)|+\sum_{i=1}^np_i \frac{\eps}{3} +\frac{\eps}{3}\nonumber\\
        &\le \eps, \nonumber
    \end{align}
    where the last line follows from \cond{tail-bound}.
\paragraph{Complexity analysis.} 
By \thm{nd-svd}, each $D_j$ uses $O(\varphi_j^{-1}\log(m\mathcal{B}_j/\eps))$ queries (simplifying the log term to $\mathcal{B}_j$ is safe). 
By \thm{qsvt}, each $W_j$ uses $\deg(P_j) = O(\varphi_{j+1}^{-1}\alpha(\eps,m)\log(m\mathcal{B}_j/\eps))$ queries. 
Thus, each $V_j$ uses $O(\varphi_{j+1}^{-1}\alpha(\eps,m)\log(m\mathcal{B}_j/\eps))$ queries.
The squared norm $\|\Pi_j V_j \ket{\Psi_{\mathrm{in}}}\|^2$ in the $j$-th loop satisfies 
\begin{align*}
    \|\Pi_j V_j \ket{\Psi_{\mathrm{in}}}\|^2 & \le \sum_{i=1}^n p_i |\beta_{i, j}|^2 P_j^2\left(\frac{\sqrt{p_i}}{2}\right) + \frac{\eps}{6m\max\{B_j, B_{j+1}\}} \\
    & \le \sum_{i\in S_j} p_i  + \sum_{i\notin S_j} p_i  \frac{\eps}{3m\max\{B_j, B_{j+1}\}} +\frac{\eps}{6m\max\{B_j, B_{j+1}\}}\\
    &\le  \varphi_{j-1}^2 |S_j|+\frac{\eps}{2m\max\{B_j, B_{j+1}\}},
\end{align*}
where $S_j = \{i\in [n]:\sqrt{p_i}\in (\varphi_{j+1}, \varphi_{j-1}]\}$, and the second line follows from $|\beta_{i,j}|\le 1$, the localization bound of $\beta_{i,j}$, and \cond{bound-P-j-1}. 
Denoting the additive error by $\delta_j = \eps/(m\max\{B_j, B_{j+1}\})$, by \cor{two-stage-ae}, the query complexity of the $j$-th loop is 
\begin{align*}
    &O\bigg( \frac{\alpha(\eps,m)\log\left(\frac{m\mathcal{B}_j}{\eps}\right)}{\varphi_{j+1}}\frac{\|\Pi_j V_j \ket{\Psi_{\mathrm{in}}}\| + \sqrt{\delta_j}}{\delta_j}\log(m)\bigg)  \\
    = {} & O\bigg( \frac{\log(m)\alpha(\eps,m)\log\left(\frac{m\mathcal{B}_j}{\eps}\right)}{\varphi_{j+1}}\frac{\sqrt{\varphi_{j-1}^2 |S_j|+\delta_j} + \sqrt{\delta_j}}{\delta_j}\bigg) \\
    = {} & O\bigg( \frac{\log(m)\alpha(\eps,m)\log\left(\frac{m\mathcal{B}_j}{\eps}\right)}{\varphi_{j+1}} \bigg( \frac{\varphi_{j-1}\sqrt{|S_j|}}{\delta_j} + \frac{1}{\sqrt{\delta_j}} \bigg) \bigg) \\
    ={}& O\bigg( \log(m)\alpha(\eps,m)\log\left(\frac{m\mathcal{B}_j}{\eps}\right) \bigg( \frac{\varphi_{j-1}}{\varphi_{j+1}}\frac{m\max\{B_j, B_{j+1}\}\sqrt{|S_j|}}{\eps} + \frac{1}{\varphi_{j+1}}\sqrt{\frac{m\max\{B_j, B_{j+1}\}}{\eps}} \bigg) \bigg)\\
    ={}& O\bigg( \log(m)\alpha(\eps,m)\log\left(\frac{m\mathcal{B}_j}{\eps}\right)\bigg(\frac{m\mathcal{B}_j\sqrt{|S_j|}}{\eps} + \frac{1}{\varphi_{j+1}}\sqrt{\frac{m\mathcal{B}_j}{\eps}}\bigg)\bigg)
\end{align*}
where the last equality holds because $\varphi_{j-1}/\varphi_{j+1} = \rho^2$ is a constant, and we bound $\max\{B_j, B_{j+1}\} \le \mathcal{B}_j$.
Recall that we define $n_j =|S_j|$, and taking the sum over $j\in[m]$ completes the proof.
\end{proof}

\subsection{An example: Shannon entropy estimation}
For $g(x) = -\log(x)$, the functional $\sum_{i=1}^n p_i g(p_i) = \mathrm{H}(p)$ corresponds to the Shannon entropy. Applying \thm{multi-level-ae-1} yields the following result. 
\begin{corollary}\label{cor:Shannon-entropy}
Given an $n$-dimensional probability distribution $p$ and its unitary oracle $\mathcal{O}_p$. There is a quantum algorithm that computes an estimate of the Shannon entropy $\mathrm{H}(p)$ to within additive error $\varepsilon$, using 
\begin{align*}
    O\left(\frac{\sqrt{n}}{\eps}\log^{4.5}\left(\frac{n}{\eps}\right)\right)
\end{align*}
queries to $\mathcal{O}_p$.
\end{corollary}
\begin{proof}
    Let $\rho = 2$, which implies $\varphi_j = 2^{-j}$. The fourth condition in \thm{multi-level-ae-1} is satisfied by choosing $m$ such that $2^m/\sqrt{m} = \Theta(\sqrt{n/\eps})$, which ensures that the tail sum is bounded by
    \begin{align*}
        \sum_{i\in[n]:\sqrt{p_i}\in [0,\varphi_m]} p_i \log\left(\frac{1}{p_i}\right) = O(n m2^{-2m} ) = O(\eps).
    \end{align*}
    Following the polynomial approximation of $\sqrt{\log(1/x)}$ in \cite[Lemma~3.3]{wang2024new}, there exist polynomials $P_j$ with bound $\mathcal{B}_j = O(j)$ and degree $\deg(P_j) = O(2^j \log(1/\eps))$ satisfying the first three conditions in \thm{multi-level-ae-1}. Then, the query complexity is \begin{align*}
    &O\bigg( \sum_{j=1}^m\log(m)\log\left(\frac{m\mathcal{B}_j}{\eps}\right)\bigg(\frac{m\mathcal{B}_j\sqrt{n_j}}{\eps} + \frac{1}{\varphi_{j+1}}\sqrt{\frac{m\mathcal{B}_j}{\eps}}\bigg) \bigg)\\
     ={} & O\bigg( \sum_{j=1}^m\log^3\left(\frac{n}{\eps}\right)\bigg(\frac{\log\left(\frac{n}{\eps}\right)\sqrt{n_j}}{\eps} + 2^j\frac{1}{\sqrt{{\eps}}}\bigg) \bigg) \\
    ={} & O\bigg( \log^3\left(\frac{n}{\eps}\right)\bigg(\frac{\log\left(\frac{n}{\eps}\right)\sqrt{mn}}{\eps} + 2^m\frac{1}{\sqrt{{\eps}}}\bigg) \bigg)  \\
    ={} &O\Big(\log^{4.5}\left(\frac{n}{\eps}\right)\frac{\sqrt{n}}{\eps}\Big),
\end{align*}
where the third line follows from the Cauchy-Schwarz inequality: 
\[
\sum_{j=1}^m \sqrt{n_j} \le \sqrt{\sum_{j=1}^m n_j}\sqrt{\sum_{j=1}^m 1} = \sqrt{nm}.
\]
\end{proof}

\section{Tsallis Entropy Estimation}

\subsection{The case of \texorpdfstring{$0<q<1$}{0<q<1}}

In the following, we show that there exists a suitable choice of parameters for estimating the functional $\sum_j p_j^q$ that satisfies the requirements of multi-level amplitude estimation when $0< q< 1$. 
\begin{proposition}\label{prop:1312130}
For $0<q<1$, let $\rho = 2$, $C = 2^{3-q}$, $m = \ceil{\frac{1}{2q} \log_{2}\rbra{\frac{3\rbra{1+2^{3-q}}n}{\varepsilon}}}$. Let $\varphi_j = \rho^{-j}$, $B_j = 4 \rho^{2j\rbra{1-q}}$, $\varepsilon_j = \frac{\varepsilon}{96\rho^{2\rbra{j+1}\rbra{1-q}}m}$, $\delta_j = 2^{-1} \rho^{-j-1}$ for $j = 0, 1, \dots, m+1$. 
Set $g(x) = x^{q-1}$.
Let $P_j\rbra{x} \coloneqq P_{1-q, \delta_j, \varepsilon_j}$ be the polynomials specified in \lem{poly-approx-negative-powers} such that
\begin{equation}\label{eq:1242050}
\begin{split}
    \abs{ P_j\rbra{x} - \frac{\delta_j^{1-q}}{2}x^{q-1} } \leq \varepsilon_j\quad \forall x \in \interval{\delta_j}{1},  \\
    \abs{ P_j\rbra{x} } \leq 1\quad \forall x \in \interval{-1}{1}, 
\end{split}
\end{equation}
with $\deg\rbra{P_j} = O\rbra{\frac{1}{\delta_j}\log\rbra{\frac{1}{\varepsilon_j}}}$. 
Then, we have the following properties.
\begin{enumerate}
    \item $B_j \ge 2\max_{x\in (\varphi_j ,\varphi_{j-1}]} |g(x^2)|$ for $j\in [m+1]$, 
    \item For all $j\in [m]$ and $x\in [-1,1]$,
    \[|P_j(x)|\leq 1.\]
    \item For all $j\in [m]$ and $x \in (\varphi_{j+1}, \varphi_{j-1}]$, 
    \[\abs{\max\cbra{B_j, B_{j+1}}P_j^2\left(\frac{x}{2}\right)-g(x^2)} \le \frac{\eps}{12m}.\]
    \item For $x \in (\varphi_{m+1},\varphi_m]$, 
    \[B_m P_m^2\rbra*{\frac{x}{2}} \leq C \abs{g\rbra{x^2}}.\]
    \item The functional satisfies 
    \[\sum_{i\in [n]: \sqrt{p_i}\in [0, \varphi_m]} p_i|g(p_i)|\le \frac{\eps}{3(1+C)}.\]
\end{enumerate}
\end{proposition}
\begin{proof}
We can verify the first and second properties directly from the definition.

For the third property,
\begin{align}
\left|\max\{B_j,B_{j+1}\}P_j^2\left(\frac{x}{2}\right)-g(x^2)\right| ={} & \left|4\cdot 2^{2(j+1)(1-q)}P^2_j\left(\frac{x}{2}\right)-x^{2(q-1)}\right| \nonumber \\
={} & 4\cdot 2^{2(j+1)(1-q)}\left|P^2_j\left(\frac{x}{2}\right)- \frac{2^{-2(j+2)(1-q)}}{2^{2}}\left(\frac{x}{2}\right)^{2(q-1)}\right| \nonumber \\
={} &4\cdot 2^{2(j+1)(1-q)}\left|P^2_j\left(\frac{x}{2}\right)- \frac{\delta_j^{2(1-q)}}{4}\left(\frac{x}{2}\right)^{2(q-1)}\right| \nonumber \\
\leq{}  &4\cdot 2^{2(j+1)(1-q)} \cdot \varepsilon_j \cdot 2 \label{eq:1242051} \\
= {} & 8\cdot 2^{2(j+1)(1-q)} \cdot \frac{\varepsilon}{96\cdot 2^{2(j+1)(1-q)}m} \nonumber \\
={} & \frac{\varepsilon}{12m}  \nonumber
\end{align}
where \eq{1242051} is due to \eq{1242050}.

For the fourth property,
\begin{align}
B_m P_m^2\left(\frac{x}{2}\right)&= 4\cdot 2^{2m(1-q)} P_m^2\left(\frac{x}{2}\right) \nonumber \\
&\leq 4\cdot 2^{2m(1-q)} \cdot \left(\frac{\delta_m^{1-q}}{2^q}x^{q-1}+\varepsilon_m\right)^2 \nonumber \\
&=4\cdot 2^{2m(1-q)}\cdot \left(\frac{2^{-(m+2)(1-q)}}{2^q}x^{q-1}+\frac{\varepsilon}{96 \cdot 2^{2(m+1)(1-q)}m}\right)^2 \nonumber \\
&\leq \left(\frac{1}{2^{1-q}}x^{q-1}+\frac{\varepsilon}{48 \cdot 2^{2(m+1)(1-q)}m}x^{q-1}\right)^2\label{eq:1251407}\\
& \leq (x^{q-1}+x^{q-1})^2\nonumber \\
&\leq C |g(x^2)|.\nonumber 
\end{align}
where \eq{1251407} is due to $x\leq \varphi_m$. 

For the fifth property,
\begin{align}
\sum_{i\in [n]: \sqrt{p_i}\in [0, \varphi_m]} p_i|g(p_i)|&=\sum_{i\in [n]: \sqrt{p_i}\in [0, \varphi_m]} p_i^q \leq n \varphi_m^{2q} =n2^{-2mq} \nonumber \\
&\leq n \cdot \frac{\varepsilon}{3(1+2^{3-q})n} \nonumber \\
&=\frac{\varepsilon}{3(1+C)}.\nonumber 
\end{align}
\end{proof}

Then, we can prove the following result.
\begin{theorem}\label{thm:1312127}
    For a constant $q \in (0,1)$, and a probability distribution $p$ of size $n$. 
    There is a quantum algorithm that, given 
    a unitary oracle $\mathcal{O}_p$ to $p$, computes an estimate of $\TsaF{q}{p}$ to within additive error $\varepsilon$, using 
    \[O\left(\frac{n^{\frac{2-q}{2q}}}{\varepsilon^{\frac{1}{q}}}\log^2\!\left(\frac{n}{\varepsilon}\right)\log\!\left(\log\!\left(\frac{n}{\varepsilon}\right)\right)\right)\]
    queries to $\mathcal{O}_p$.
\end{theorem}
\begin{proof}
In \prop{1312130}, we showed that the chosen parameters satisfy the constraints required by \thm{multi-level-ae-1}. Thus, we can obtain an estimate of $\TsaF{q}{p}=\sum_{i=1}^n p_i^{q}$ with query complexity: 
\begin{align}
& \sum_{j=1}^m {\log\rbra{m}\log\left(\frac{m\mathcal{B}_j}{\eps}\right)}\left(\frac{m\mathcal{B}_j\sqrt{|\{i\in [n]:\sqrt{p_i}\in (\varphi_{j+1}, \varphi_{j-1}]\}|}}{\eps} + \frac{1}{\varphi_{j+1}}\sqrt{\frac{m\mathcal{B}_j}{\varepsilon}} \right)\nonumber \\
\leq {}  & \log(m)\log\left(\frac{m\mathcal{B}_m}{\varepsilon}\right) \sum_{j=1}^m \left(\frac{m\mathcal{B}_j\sqrt{|\{i\in [n]:\sqrt{p_i}\in (\varphi_{j+1}, \varphi_{j-1}]\}|}}{\varepsilon}+\frac{1}{\varphi_{j+1}}\sqrt{\frac{m\mathcal{B}_j}{\varepsilon}}\right) \nonumber \\
\leq {} &\log(m)\log\left(\frac{m\mathcal{B}_m}{\varepsilon}\right)\left(\frac{m}{\varepsilon}\sqrt{\sum_{j=1}^m \mathcal{B}_j^2} \sqrt{\sum_{j=1}^m|\{i\in [n]:\sqrt{p_i}\in (\varphi_{j+1}, \varphi_{j-1}]\}|}+\sum_{j=1}^m\frac{1}{\varphi_{j+1}}\sqrt{\frac{m\mathcal{B}_j}{\varepsilon}}\right)\nonumber \\
\leq {} &\log(m)\log\left(\frac{m\mathcal{B}_m}{\varepsilon}\right)\left(\frac{m}{\varepsilon}\sqrt{\sum_{j=1}^m 16\cdot 2^{4(j+2)(1-q)}} \sqrt{2n}+\sqrt{\frac{m}{\varepsilon}}\sum_{j=1}^m 2^{j+1}\cdot 2\cdot 2^{(j+2)(1-q)} \right)\nonumber\\
= {} &O\left(\log(m)\log\left(\frac{m\mathcal{B}_m}{\varepsilon}\right) \left(\frac{m}{\varepsilon}\cdot 2^{2m(1-q)} \cdot\sqrt{n} + \sqrt{\frac{m}{\varepsilon}}\cdot 2^{m(2-q)}\right)\right) \nonumber\\
={} &O\left(\log(m)\log\left(\frac{m\mathcal{B}_m}{\varepsilon}\right) \left(\frac{m}{\varepsilon}\cdot \left(\frac{3(1+2^{3-q})n}{\varepsilon}\right)^{\frac{1-q}{q}} \cdot\sqrt{n} + \sqrt{\frac{m}{\varepsilon}}\cdot \left(\frac{3(1+2^{3-q})n}{\varepsilon}\right)^{\frac{2-q}{2q}}\right)\right) \nonumber\\
={} &O\left(\log(m)\log\left(\frac{m\mathcal{B}_m}{\varepsilon}\right) \left(\frac{m}{\varepsilon} \cdot \left(\frac{n}{\varepsilon}\right)^{\frac{1-q}{q}}\cdot \sqrt{n}+\sqrt{\frac{m}{\varepsilon}}\cdot \left(\frac{n}{\varepsilon}\right)^{\frac{2-q}{2q}}\right)\right) \nonumber \\
={} &O\left(\log(m)\log(m\mathcal{B}_m/\varepsilon) \left(m\frac{n^{\frac{2-q}{2q}}}{\varepsilon^{\frac{1}{q}}}+\sqrt{m}\cdot \frac{n^{\frac{2-q}{2q}}}{\varepsilon^{\frac{1}{q}}}\right)\right)\nonumber\\
={} & O\left(\frac{n^{\frac{2-q}{2q}}}{\varepsilon^{\frac{1}{q}}}\log^2\!\left(\frac{n}{\varepsilon}\right)\log\!\left(\log\!\left(\frac{n}{\varepsilon}\right)\right)\right).\nonumber
\end{align}
Note that $\Tsa{q}{p} = \frac{1}{1-q}(\TsaF{q}{p}-1)$, 
the complexity is the same for obtaining an $\varepsilon$-estimate of $\Tsa{q}{p}$.
\end{proof}

\subsection{The case of \texorpdfstring{$q>1$}{q>1}}

In the following, we show that there exists a suitable choice of parameters for estimating the functional $\sum_j p_j^q$ for $q>1$ that satisfies the requirements of multi-level amplitude estimation.

\begin{proposition}\label{prop:tsallis-q-greater-poly-approx}
    Let $q > 1$ be some constant, $\varepsilon \in\interval[open]{0}{1/2}$, $m = \ceil{\frac{\log_2 (9/\varepsilon)}{2(q-1)}}$, $g(x)\coloneqq x^{q-1}$, $\varphi_j = 2^{-j}$ for $j = 0, \ldots, m+1$, $\varphi_{m+2} =0$,  $B_j = 2^{2-2(j-2)(q-1)}$ for $j = 1, \ldots , m+1$, and $\varepsilon_j = \frac{\varepsilon}{24m B_j}$ for $j=1, \ldots, m-1$ and $\varepsilon_m = \min\cbra{\frac{\varepsilon}{24m B_m}, \frac{\varepsilon}{ 2^{1+(m+2)(q-1)}B_m}}$. Let $P_j(x) \coloneqq S_{q-1, \varphi_{j+1}/2, \varphi_{j-1}/2, \varepsilon_j}$ for $j\in [m]$, where $S$ is the polynomial specified in~\lem{poly-approx-positive-powers}. Then, it holds that
       \begin{enumerate}
    \item  For all $j \in [m]$ and $x\in[-1,1]$,
    \begin{equation*}
         |P_j(x)|\le 1. 
    \end{equation*}
    \item For all $j \in [m+1]$, 
    \begin{equation*}
    B_j \ge 2\max_{x\in (\varphi_j ,\varphi_{j-1}]} |g(x^2)|.
    \end{equation*}
    \item  For all $j \in [m]$ and $x \in (\varphi_{j+1}, \varphi_{j-1}]$,
    \begin{equation*}
        \left|\max\{B_j, B_{j+1}\}P_j^2\left(\frac{x}{2}\right)-g(x^2)\right| \le \frac{\eps}{12m}. 
    \end{equation*}
    \item For any $x\in (\varphi_{m+1},\varphi_m]$,
    \begin{equation*}
         B_m P_m^2\left(\frac{x}{2}\right) \le 2|g(x^2)|. 
    \end{equation*}
    \item The functional satisfies
    \begin{equation*}
        \sum_{i\in [n]: \sqrt{p_i}\in [0,\varphi_m]} p_i|g(p_i)|\le \frac{\eps}{9}.
    \end{equation*}
    \item $P_1, \ldots, P_m\in \mathbb{R}[x]$ are of the same parity with $\deg(P_j) = O(\log(mB_j/\eps)\log (1/\varepsilon)/\varphi_{j+1})$ for $j = 1, 2, \ldots, m-1$, and $\deg(P_m) = O(\log(mB_m/\eps)\log^2 (1/\varepsilon)/\varphi_{m+1})$.
\end{enumerate}
\end{proposition}

\begin{proof}
    We first show the parameters satisfy the requirement of~\lem{poly-approx-positive-powers}. Specifically, noting that $B_j > B_{j+1}$, we have
    \[
    \varepsilon_j \le \frac{\varepsilon}{24 m B_m}\le \frac{9}{24\log \rbra*{9/\varepsilon}}\le \frac{1}{2}.
    \]

    For (1), this is direct by~\lem{poly-approx-positive-powers}.

    For (2), note that $2\max_{x\in (\varphi_j ,\varphi_{j-1}]} |g(x^2)| = 2 \varphi_{j-1}^{2(q-1)} = 2^{1-2(q-1)(j-1)}$. We have 
    \begin{equation*}
    B_j = 2^{2-2(j-2)(q-1)}\ge 2^{1-2(q-1)(j-1)}.
    \end{equation*}

    For (3), noting that $B_j > B_{j+1}$, we have $\forall x\in \interval{\frac{\varphi_{j+1}}{2}}{\frac{\varphi_{j-1}}{2}}$, 
    \begin{equation*}
        \abs{P_j(x)- g_j(x)}\le \varepsilon_j, 
    \end{equation*}
    where $g_j(x)\coloneqq\frac{x^{q-1}}{2 \varphi_{j-1}^{q-1}} = \frac{x^{q-1}}{2^{1-(q-1)(j-1)}}$.
    For $x\in \interval{\varphi_{j+1}/2}{\varphi_{j-1}/2}$, $\abs{g(x)}\le 1$. Therefore, for $x \in (\varphi_{j+1}, \varphi_{j-1}]$, we have
    \begin{equation*}
        \abs{P_j^2\rbra*{\frac{x}{2}}-g_j^2\rbra*{\frac{x}{2}}} \le \abs{P_j\rbra*{\frac{x}{2}}+g_j\rbra*{\frac{x}{2}}} \abs{P_j\rbra*{\frac{x}{2}}-g_j\rbra*{\frac{x}{2}}} \le 2\varepsilon_j.
    \end{equation*}
    Since $B_j g_j^2\rbra{x/2} = 2^{2-2(j-2)(q-1)} \cdot \frac{x^{2(q-1)}}{2^{2(q-1)}2^{2-2(q-1)(j-1)}} = x^{2(q-1)} = g(x^2) $, we know
    \begin{equation*}
        \abs{\max\cbra{B_j,B_{j+1}}P_j^2\rbra*{\frac{x}{2}}-g(x^2)} = 
        \abs{B_jP_j^2\rbra*{\frac{x}{2}}-B_jg^2\left(\frac{x}{2}\right)} \le 2B_j\varepsilon_j \le \frac{\varepsilon}{12m}.
    \end{equation*}

    For (4), noting that for $x\in \interval[open left]{\varphi_{m+1}/2}{\varphi_{m-1}/2}$, we have 
    \begin{equation*}
        \abs{P_m^2\rbra*{\frac{x}{2}}-g_m^2\rbra*{\frac{x}{2}}} \le 2\varepsilon_m.
    \end{equation*}
    This gives
    \begin{equation*}
        \abs{P_m^2\rbra*{\frac{x}{2}}} \le g_m^2\rbra*{\frac{x}{2}} +   2\varepsilon_m.
    \end{equation*}
    Using $B_m g^2_m(x/2) = g(x^2)$, we only need to bound $2B_m \varepsilon_m$. Note that for $x\in \interval[open left]{\varphi_{m+1}/2}{\varphi_{m-1}/2}$, $g(x^2)\ge g(\varphi_{m+1}/2)^2 = 2^{-(m+2)(q-1)}$. By the choice of $\varepsilon_m$, it holds that $2B_m \varepsilon_m\le \frac{1}{2^{1+(m+2)(q-1)}} \le g(x^2)$.
    Therefore, 
    \begin{equation*}
       B_mP_m^2\rbra*{\frac{x}{2}} \le B_mg_m^2\rbra*{\frac{x}{2}} +   2B_m\varepsilon_m \le 2g(x^2).
    \end{equation*}

    For (5), 
    \begin{equation*}
          \sum_{i\in [n]: \sqrt{p_i}\in [0,\varphi_m]} p_i|g(p_i)|\le f(\varphi_m^2) \rbra*{\sum_{i=1}^n p_i } \le \varphi_m^{2(q-1)}\le \frac{\varepsilon}{9}.
    \end{equation*}

    For (6), by~\lem{poly-approx-positive-powers}, we have
    \begin{equation*}
        \deg(P_j) = O\rbra*{\frac{1}{\varphi_{j+1}}\log \rbra*{\frac{1}{\varphi_j \varphi_{j+1}\varepsilon_j}}} = O\rbra*{\frac{j}{\varphi_{j+1}}\log \rbra*{\frac{mB_j}{\varepsilon}}} = O\rbra*{\frac{\log \rbra*{\frac{1}{\varepsilon}}}{\varphi_{j+1}}\log \rbra*{\frac{mB_j}{\varepsilon}}}.
    \end{equation*}
   Also, we have that $\log\rbra{1/\varepsilon_m} = O\rbra*{m\log\rbra{B_m/\varepsilon}} = O\rbra*{\log\rbra{1/\varepsilon}\log\rbra{mB_m/\varepsilon}}$.

    Combining the above, we have finished our proof.
\end{proof}

Applying~\thm{multi-level-ae-1}, we have the following theorem.
\begin{theorem}\label{thm:Tsallis-entropy-q-greater-than-one}
    For a constant $q > 1$, and a probability distribution $p$ of size $n$. 
    There is a quantum algorithm that, given
    a purified quantum query access $\mathcal{O}_p$ to $p$, computes an estimate of $\Tsa{q}{p}$ to within additive error $\varepsilon$, using 
\begin{equation*}
    \begin{cases}
      O\bigg( \frac{1}{\varepsilon}\log^{4}\rbra*{\frac{1}{\varepsilon}} \log\rbra*{\log\rbra*{\frac{1}{\varepsilon}}}  \bigg), & \text{ if } q > 1.5 \\
           O\bigg( \frac{1}{\varepsilon}\log^{5}\rbra*{\frac{1}{\varepsilon}} \log\rbra*{\log\rbra*{\frac{1}{\varepsilon}}}  \bigg), & \text{ if } q = 1.5  \\
               O\bigg( \frac{1}{\varepsilon^{\frac{1}{2(q-1)}}}\log^{4}\rbra*{\frac{1}{\varepsilon}} \log\rbra*{\log\rbra*{\frac{1}{\varepsilon}}}  \bigg) & \text{ if } 1< q < 1.5
\end{cases}
\end{equation*}
    queries to $\mathcal{O}_p$.
\end{theorem}

\begin{proof}
In~\prop{tsallis-q-greater-poly-approx}, we showed that the chosen parameters satisfy the constraints required by~\thm{multi-level-ae-1}.
Therefore,  we can obtain an estimate of $\TsaF{q}{p}=\sum_{i=1}^n p_i^{q}$ with query complexity:
\begin{equation*}
    \begin{aligned}
               & O\bigg( \log^{2}\rbra*{\frac{1}{\varepsilon}}\sum_{j=1}^m\log(m)\log\rbra*{\frac{m\mathcal{B}_j}{\eps}}\bigg(\frac{m\mathcal{B}_j\sqrt{n_j}}{\eps} + \frac{1}{\varphi_{j+1}}\sqrt{\frac{m\mathcal{B}_j}{\eps}}\bigg) \bigg) \\
               &= O\bigg( \log^{2}\rbra*{\frac{1}{\varepsilon}} \log(m)\log\rbra*{\frac{m\mathcal{B}_j}{\eps}} \sum_{j=1}^m\bigg(\frac{m\mathcal{B}_j\sqrt{n_j}}{\eps} + \frac{1}{\varphi_{j+1}}\sqrt{\frac{m\mathcal{B}_j}{\eps}}\bigg) \bigg) \\
                &= O\bigg( \log^{3}\rbra*{\frac{1}{\varepsilon}} \log\rbra*{\log\rbra*{\frac{1}{\varepsilon}}} \bigg(\frac{m}{\eps} \sum_{j=1}^m  \mathcal{B}_j\sqrt{n_j}+ \sum_{j=1}^m\frac{1}{\varphi_{j+1}}\sqrt{\frac{m\mathcal{B}_j}{\eps}}\bigg) \bigg). 
    \end{aligned}
\end{equation*}
Note that $n_j \varphi_{j-1}^2 \le 1$, we have 
\begin{equation*}
    \frac{m}{\varepsilon}\sum_{j=1}^m  \mathcal{B}_j\sqrt{n_j} = O\rbra*{ \sum_{j=1}^m   \frac{m}{\varepsilon}\frac{B_j}{\varphi_j}}=  O\rbra*{ \frac{m}{\varepsilon} \sum_{j=1}^m 2^{(3-2q)j}}.
\end{equation*}
Also, for our choice of $B_j$, we have 
$\frac{mB_j}{\varepsilon}\ge \frac{mB_m}{\varepsilon}\ge \frac{\log\rbra{9/\varepsilon}}{9} \ge \Omega(1)$, giving 
\begin{equation*}
     \sum_{j=1}^m\frac{1}{\varphi_{j+1}}\sqrt{\frac{m\mathcal{B}_j}{\eps}} = O\rbra*{\sum_{j=1}^m\frac{mB_j}{\varepsilon \varphi_j }}.
\end{equation*}

For $q >1.5$, 
\begin{equation*}
   O\rbra*{\sum_{j=1}^m\frac{mB_j}{\varepsilon \varphi_j }} =    O\rbra*{ \frac{m}{\varepsilon} \sum_{j=1}^m 2^{(3-2q)j}} =  O\rbra*{ \frac{m}{\varepsilon} }, 
\end{equation*}
and the final complexity is 
\begin{equation*}
    O\bigg( \frac{1}{\varepsilon}\log^{4}\rbra*{\frac{1}{\varepsilon}} \log\rbra*{\log\rbra*{\frac{1}{\varepsilon}}}  \bigg).
\end{equation*}

For $q =1.5$, 
\begin{equation*}
     O\rbra*{\sum_{j=1}^m\frac{mB_j}{\varepsilon \varphi_j }}=  O\rbra*{ \frac{m}{\varepsilon} \sum_{j=1}^m 2^{(3-2q)j}} =  O\rbra*{ \frac{m^2}{\varepsilon} }, 
\end{equation*}
and the final complexity is 
\begin{equation*}
    O\bigg( \frac{1}{\varepsilon}\log^{5}\rbra*{\frac{1}{\varepsilon}} \log\rbra*{\log\rbra*{\frac{1}{\varepsilon}}}  \bigg).
\end{equation*}

For $1< q < 1.5$,
\begin{equation*}
    O\rbra*{\sum_{j=1}^m\frac{mB_j}{\varepsilon \varphi_j }}=  O\rbra*{ \frac{m}{\varepsilon} \sum_{j=1}^m 2^{(3-2q)j}} =  O\rbra*{ \frac{m2^{(3-2q)m}}{\varepsilon} }, 
\end{equation*}
and the final complexity is 
\begin{equation*}
    O\bigg( \frac{1}{\varepsilon^{\frac{1}{2(q-1)}}}\log^{4}\rbra*{\frac{1}{\varepsilon}} \log\rbra*{\log\rbra*{\frac{1}{\varepsilon}}}  \bigg).
\end{equation*}
Note that $\Tsa{q}{p} = \frac{1}{1-q}(\TsaF{q}{p}-1)$, 
the complexity is the same for obtaining an $\varepsilon$-estimate of $\Tsa{q}{p}$.
\end{proof}

\section{Lower Bounds}
In this section, we show quantum query lower bounds for estimating Tsallis entropies given purified quantum access to the probability distributions.

\begin{theorem}\label{thm:lower-bound-q-tsallis}
    Any quantum query algorithm for estimating the $q$-Tsallis entropy of an unknown probability distribution of size $n$ to within additive error $\varepsilon$ requires query complexity:
    \begin{itemize}
        \item $\Omega\left(\frac{n^{\frac{1}{2q}}}{\varepsilon^{\frac{1}{2q}}\sqrt{\log(n/\varepsilon)}}\right)$ for $0 < q < 0.5$ and $\varepsilon \geq \Omega(n^{\frac{2q^2-2q+1}{1-2q}}\log^{\frac{q}{1-2q}}(n))$,
        \item $\Omega\left(\frac{n^{\frac{1}{2q}}}{\varepsilon^{\frac{1}{2q}}\sqrt{\log(n/\varepsilon)}}+\frac{n^{1-q}}{\varepsilon}\right)$ for $0.5 \leq q < 1$,
        \item $\Omega\left(\frac{1}{\varepsilon^{\frac{1}{2(q-1)}}\sqrt{\log(1/\varepsilon)}}\right)$ for $1 < q < 1.5$,
        \item $\Omega\left(\frac{1}{\varepsilon}\right)$ for $q \geq 1.5$.
    \end{itemize}
\end{theorem}

\begin{proof}
    Our quantum query complexity lower bounds are obtained by the quantum sample-to-query lifting method \cite{WZ25a,WZ25b,TWZ25,CWZ25}. 
    Specifically, as mentioned in \cite[Section 2.1]{CWZ25}, if testing a property of an unknown probability distribution requires sample complexity $\Omega(S)$, then this problem requires quantum query complexity $\Omega(\sqrt{S})$ in the purified quantum query access model. 
    Therefore, in the remainder of this proof, we focus on the sample complexity lower bounds for estimating the $q$-Tsallis entropy for different ranges of $q$ separately. 

    \begin{itemize}
        \item For $0 < q < 0.5$, in \cite[Theorem 2]{JVHW15}, it was shown that the minimax mean squared error for estimating the $q$-Tsallis entropy is $\Theta\left(\frac{n^2}{(S\log(S))^{2q}}\right)$ for $n \geq \Omega(S^{1-\frac{1}{2q}}\log(S))$ (with the specific constraints stated in \cite[Page 2852]{JVHW15}). 
        To ensure the additive error $\varepsilon$, we need to set $\varepsilon^2 = \Theta\left(\frac{n^2}{(S\log(S))^{2q}}\right)$ for $n \geq \Omega(S^{1-\frac{1}{2q}}\log(S))$, which gives $S \geq \Omega\left(\frac{n^{\frac{1}{q}}}{\varepsilon^{\frac{1}{q}}\log(n/\varepsilon)}\right)$ with the constraints $\varepsilon \geq \Omega(n^{\frac{2q^2-2q+1}{1-2q}}\log^{\frac{q}{1-2q}}(n))$. 

        \item For $0.5 \leq q < 1$, in \cite[Theorem 2]{JVHW15}, it was shown that the minimax mean squared error for estimating the $q$-Tsallis entropy is $\Theta\left(\frac{n^2}{(S\log(S))^{2q}}+\frac{n^{2-2q}}{S}\right)$.\footnote{The case of $q = 0.5$ should be classified under \cite[Theorem 2(2)]{JVHW15} rather than \cite[Theorem 2(1)]{JVHW15}.} 
        To ensure the additive error $\varepsilon$, we need to set $\varepsilon^2 = \Theta\left(\frac{n^2}{(S\log(S))^{2q}}+\frac{n^{2-2q}}{S}\right)$, which gives $S \geq \Omega\left(\frac{n^{\frac{1}{q}}}{\varepsilon^{\frac{1}{q}}\log(n/\varepsilon)}+\frac{n^{2-2q}}{\varepsilon^2}\right)$.

        \item For $1 < q < 1.5$, in \cite[Theorems 3 and 4]{JVHW15}, it was shown that the minimax mean squared error for estimating the $q$-Tsallis entropy is $\Theta\left(\frac{1}{(S\log(S))^{2(q-1)}}\right)$.
        To ensure the additive error $\varepsilon$, we need to set $\varepsilon^2 = \Theta\left(\frac{1}{(S\log(S))^{2(q-1)}}\right)$, which gives $S \geq \Omega\left(\frac{1}{\varepsilon^{\frac{1}{q-1}}\log(1/\varepsilon)}\right)$. 

        \item For $q \geq 1.5$, in \cite[Section VI-A]{JVHW17}, it was shown that the minimax mean squared error for estimating the $q$-Tsallis entropy is $\Theta\left(\frac{1}{S}\right)$.
        To ensure the additive error $\varepsilon$, we need to set $\varepsilon^2 = \Theta\left(\frac{1}{S}\right)$, which gives $S \geq \Omega\left(\frac{1}{\varepsilon^{2}}\right)$. 
    \end{itemize}
\end{proof}

\bibliographystyle{alphaurl}
\bibliography{ref}
\end{document}